\documentclass[twocolumn,amssymb,nobibnotes, floatfix, prx, superscriptaddress,aps,a4paper]{revtex4-2}
\usepackage[qm]{qcircuit} 
\usepackage{graphicx}

\usepackage[margin=2cm]{geometry}
\usepackage{amsmath,amssymb} 
\usepackage{dcolumn}
\usepackage{bm}
\usepackage[hidelinks]{hyperref}
\usepackage{physics} 
\usepackage[table,xcdraw,dvipsnames]{xcolor}
\usepackage{float}
\usepackage[customcolors]{hf-tikz} 
\usepackage[utf8]{inputenc}
\usepackage{enumitem}
\usepackage{booktabs}

\usepackage{listings}

\usepackage{tcolorbox}
\newtcolorbox{redbox}[1]{colback=red!5!white,colframe=red!50!white,fonttitle=\bfseries,title=#1}
\newtcolorbox{bluebox}[1]{colback=blue!5!white,colframe=blue!50!white,fonttitle=\bfseries,title=#1}
\newtcolorbox{greenbox}[1]{colback=green!5!white,colframe=green!50!white,fonttitle=\bfseries,title=#1}

\usepackage{pdfrender}
\newcommand*{\boldify}[1]{%
  \textpdfrender{%
    TextRenderingMode=FillStroke,%
    LineWidth=.60pt,%
  }{#1}%
}

\newcommand{\vbra}[1]{\bra{\boldify{#1}}}
\newcommand{\vket}[1]{\ket{\boldify{#1}}}

\newcommand{\E}{\mathcal{E}}

\usepackage{dsfont}
\newcommand{\ident}{
  \mathds{1}
}

\renewcommand{\cal}[1]{\mathcal{#1}}
\DeclareMathAlphabet\mathbfcal{OMS}{cmsy}{b}{n}

\usepackage{amsthm}
\newtheorem{theorem}{Theorem}[section]
\newtheorem{corollary}{Corollary}[section]
\newtheorem{lemma}{Lemma}[section]

\definecolor{ibm0}{HTML}{648FFF}
\definecolor{ibm1}{HTML}{785EF0}
\definecolor{ibm2}{HTML}{DC267F}
\definecolor{ibm3}{HTML}{FE6100}
\definecolor{ibm4}{HTML}{FFB000}

\newcounter{protocol}
\newenvironment{protocol}[1][]{
  \refstepcounter{protocol}
  \par\noindent\rule{\columnwidth}{1.4pt}  \vspace{-0.3cm}
  \par\noindent Protocol~\theprotocol: \textbf{#1} \rmfamily 
  \par\noindent\rule{\columnwidth}{1pt} \vspace{-0.5cm} }{\par\noindent\rule{\columnwidth}{1.4pt}\medskip}

\usepackage[normalem]{ulem}
\def\>{\rangle}
\def\<{\langle}

\def\I{ \ident }

\def\H {\mathcal{H}}
\def\B {\mathcal{B}}
\def\U {\mathcal{U}}
\def\V {\mathcal{V}}
\def\F {\mathcal{F}}
\def\D {\mathcal{D}}
\def\R {\mathcal{R}}
\def\S {\mathcal{S}}
\def\C {\mathcal{C}}

\def\G {\mathcal{G}}
\def\J {\mathcal{J}}
\def\Q {\mathcal{Q}}

\def\M {\mathcal{M}}

\def\bE {{\bar{\mathcal{E}}}}
\def\swap{SW\!\!AP}
\def\x {\otimes}

\usepackage{subcaption}
\DeclareFontFamily{U}{mathx}{}
\DeclareFontShape{U}{mathx}{m}{n}{<-> mathx10}{}
\DeclareSymbolFont{mathx}{U}{mathx}{m}{n}
\DeclareMathAccent{\widenoisy}{0}{mathx}{"71}

\newcommand{\noisy}[1]{{#1}_N}
\usepackage{ulem}
\begin{document}
\title{A simple formulation of no-cloning and no-hiding \\ that admits efficient and robust verification}
\date{\today}
\author{Matthew Girling}
\email{m.j.girling@leeds.ac.uk}
\affiliation{School of Physics and Astronomy, University of Leeds, Leeds LS2 9JT, United Kingdom}
\author{Cristina C\^{i}rstoiu}
\affiliation{Quantinuum, 13-15 Hills Road, Cambridge CB2 1NL, United Kingdom} 
\author{David Jennings}
\affiliation{School of Physics and Astronomy, University of Leeds, Leeds LS2 9JT, United Kingdom}
\affiliation{Department of Physics, Imperial College London, London SW7 2AZ, United Kingdom}

\begin{abstract}
Incompatibility is a feature of quantum theory that sets it apart from classical theory, and the inability to clone an unknown quantum state is one of the most fundamental instances. The no-hiding theorem is another such instance that arises in the context of the black-hole information paradox, and can be viewed as being dual to no-cloning.  Here, we formulate both of these fundamental features of quantum theory in a single form that is amenable to efficient verification, and that is robust to errors arising in state preparation and measurements. We extend the notion of unitarity - an average figure of merit that for quantum theory captures the coherence of a quantum channel - to general physical theories. Then, we introduce the notion of compatible unitarity pair (CUP) sets, that correspond to the allowed values of unitarities for compatible channels in the theory.  We show that a CUP-set constitutes a simple `fingerprint' of a physical theory, and that incompatibility can be studied through them. We derive information-disturbance constraints on quantum CUP-sets that encode both the no-cloning/broadcasting and no-hiding theorems of quantum theory. We then develop randomised benchmarking protocols that efficiently estimate quantum CUP-sets and provide simulations using IBMQ of the simplest instance. Finally, we discuss ways in which CUP-sets and quantum no-go theorems could provide additional information to benchmark quantum devices.
\end{abstract}
\maketitle

\section{Introduction}

Quantum physics places much stronger limits on how we can transform information, compared to classical physics. These limits can be captured by the notion of incompatibility, that encapsulates fundamental impossibility results in quantum theory~\cite{wootters1982single,yuen1986amplification,barnum2007generalized, nobroadcasting1}.  The most commonly encountered form of incompatibility refers to measurements -- position and momentum cannot be simultaneously measured with the same precision -- leading to formulations of no information without disturbance \cite{kretschmann2008information}. However, incompatibility can be described far more generally \cite{heinosaari2016invitation, bluhm2022incompatibility, jenvcova2018incompatible}. Two local processes on systems $A$ and $B$ are said to be compatible if there exists a \emph{global} process that can produce both. The no-broadcasting theorem, an extension of the famous no-cloning theorem, can be cast as the incompatibility of local identity channels at $A$ and $B$~\cite{nobroadcasting1}. It is readily seen that if a physical theory admits perfect cloning, such as classical theory, then the theory cannot have any form of incompatibility.  

Quantum technologies open new directions to experimentally test foundational aspects of quantum theory~\cite{hensen2015loophole, giustina2015significant, shalm2015strong, samal2011experimental,gao2022experimental}. However, current devices are inherently noisy with error mitigation and correction being key obstacles to overcome for scalable quantum computing~\cite{preskill2018quantum,bombin2015single,suzuki2022quantum,nielsen2002quantum}. This presents a challenge in developing tests for foundational properties in a way that is robust to errors arising from the implementation of the experiment itself (e.g at the state preparation and measurement stage). Viewed  another way, such tests can also produce valuable benchmarks of errors in noisy intermediate scale quantum devices \cite{sadana2022testing}, as they have clear operational significance rooted in fundamental properties of quantum mechanics.

With this in mind, in this work we address the following question:
\begin{quote}\centering
\textit{{Can we formulate measures of quantum incompatibility that can be robustly \& efficiently estimated?}}
\end{quote}

Existing criteria to decide the compatibility of quantum channels can be formulated generally in terms of semidefinite programming and by introducing witnesses of incompatibility \cite{carmeli2019witnessing, carmeli2019quantum}. The task has also been shown to be equivalent to the quantum state marginal problem \cite{girard2021jordan}. Other formulations rely on the Fisher metric \cite{zhang2022fisher} or on the diamond norm \cite{kretschmann2008continuity} to capture information disturbance trade-offs. Another approach is via robustness measures \cite{haapasalo2015robustness,designolle2019incompatibility}. Evaluating these different figures of merit for incompatibility requires extensive optimisations that typically assume access to a full description of the (quantum) processes involved.  

The inability to clone a quantum state \cite{wootters1982single} can be shown to be an extremal case of incompatibility~\cite{heinosaari2016invitation}, and the ability to ``hide'' data in correlations can be viewed as being dual to cloning. This problem arises in the black-hole information paradox \cite{braunstein2007quantum,giddings1995black,maldacena2020black}, and the no-hiding theorem was established to prove the impossibility of hiding a qubit state in quantum correlations \cite{braunstein2007quantum}. This has the implication that black hole information must have some degree of spatial localization, either within the black hole interior or in the external region to the black hole \cite{braunstein2007quantum,giddings1995black,maldacena2020black}.

No-cloning and no-hiding can also be related to other quantum impossibility results such as no-masking \cite{modi2018masking,zhu2021hiding} and no-deleting \cite{kumar2000impossibility}. There have been some recent experimental tests of the no-hiding theorem \cite{samal2011experimental} including the utilization of small scale quantum computers \cite{kalra2019demonstration}. No-cloning has also been tested in the context of information disturbance \cite{gao2022experimental}. However, here we develop a broader framework that exploits recent theoretical ideas that arise in the analysis of quantum technologies.
 
Our approach to this problem is motivated by randomised benchmarking techniques \cite{eisert2020quantum,kliesch2021theory}. Such methods produce estimates of average channel properties (fidelity, unitarity etc.) in a way that is robust to state preparation and measurement (SPAM) errors and does not require exponentially difficult process tomography. In particular, we argue that the unitarity of a quantum channel, which is a measure of its coherence  \cite{wallman2015estimating,dirkse2019efficient,carignan2019bounding}, is a natural means to \emph{simultaneously} describe both no-cloning and no-hiding. In particular we show how such incompatibility can be captured by unitarity within a single information-disturbance inequality.

At a high-level, our work can be viewed as extending the simple concept of the purity of a state, which is a measure of disorder \footnote{in quantum theory this is $\gamma(\rho):= \tr[\rho^2]$ for any state $\rho$}, to what can be viewed as a purity-measure of the physical theory itself. This extension serves as a simple and intuitive 2--dimensional ``fingerprint'' of the theory. An example for quantum theory is shown in Figure~\ref{fig:discrete-banana-ibm-Belem-iRB}. 

\begin{figure}[t]
    \includegraphics[width=7cm]{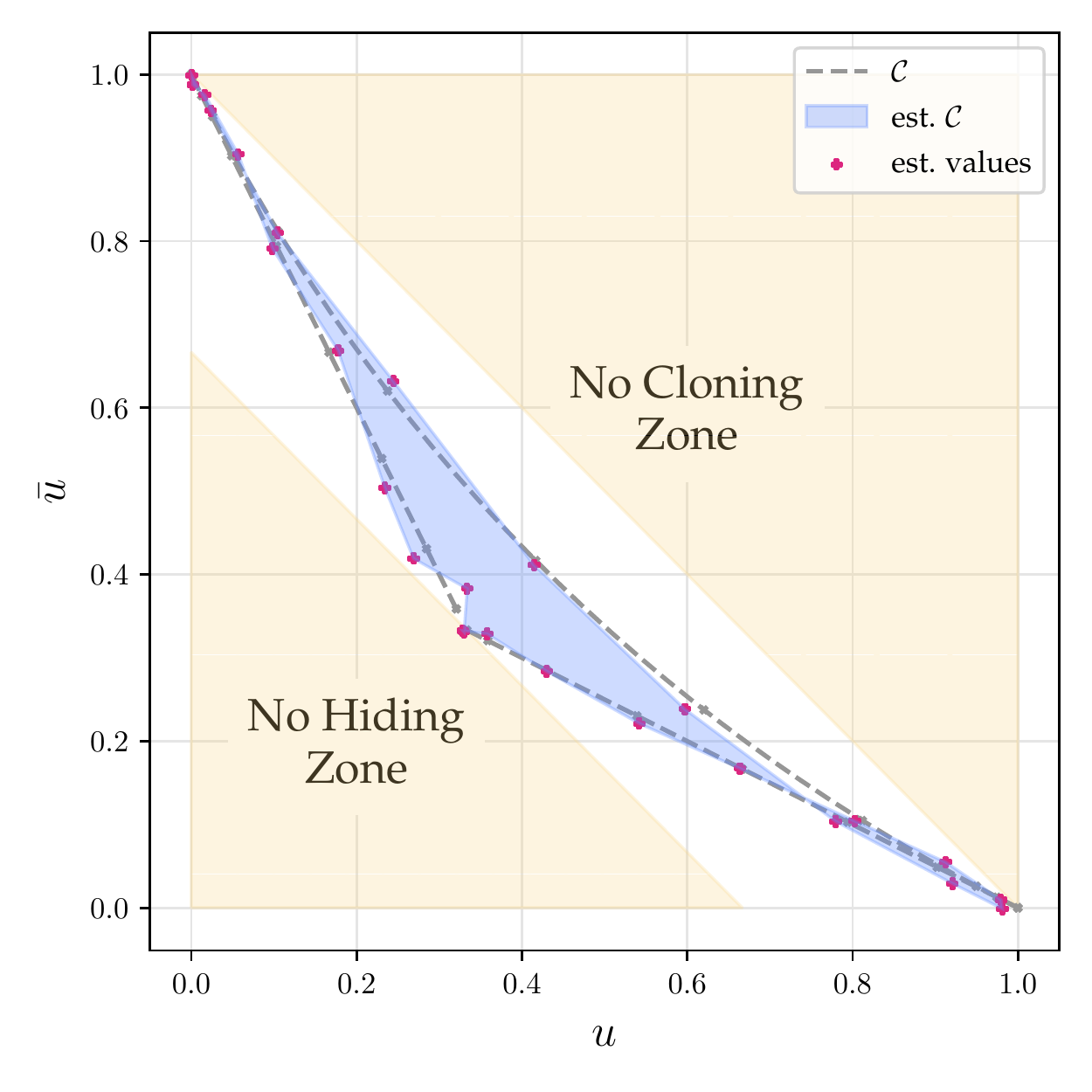}
    \caption{\textbf{Robust and efficient verification of quantum incompatibility.} In classical theory we have the ability to perfectly clone and perfectly hide classical information. In contrast, quantum theory has fundamental incompatibility that prohibits the same behaviour. This is captured by defining CUP-sets, and shown here is the estimation of the simplest quantum CUP-set $\C$. The reversible CUP-set for classical theory corresponds to the full boundary of the unit square $[0,1]^2$, and allows perfect cloning (the point $(1,1)$) and perfect hiding (the point $(0,0)$). Using benchmarking techniques we estimate $\C$, shown here, on an IBM Q device and find that it saturates the general quantum bounds we derive in Theorem~\ref{thm:cup-bounds}. Verifying such fundamental bounds provides a means to test the performance of emerging quantum computers.}
    \label{fig:discrete-banana-ibm-Belem-iRB}
\end{figure}

\subsection{Structure \& main results of the paper}

 Our main focus is to simultaneously handle both classical and quantum theories under a unifying umbrella using average channel properties.
In Section \ref{sec:compatible-unitarity-pairs}, we first develop a generalization of the \textit{unitarity} $u(\E)$ of a quantum channel, that allows extensions of our work to more general physical theories~\cite{chiribella2016quantum,plavala2021general,barnum2006cloning}. We show that the unitarity has key properties that make it well suited to capturing compatibility, compared to other average measures such as fidelities \cite{horodecki1999general,knill2008randomized}. 

We briefly summarize the framework we develop to capture incompatibility of a theory.
For any theory and a globally isometric process, $\V_{X \to AB}$, from a system $X$ into systems $AB$, we consider \emph{marginal channels}
\begin{equation}
    \begin{split}
        \E &:= tr_B \circ \V_{X \to AB}, \\
        \bar{\E} &:= tr_A \circ \V_{X \to AB},
    \end{split}
\end{equation}
by tracing out either $A$ or $B$ respectively.
These channels let us define \textit{compatible unitarity pairs} (CUPs), which we write as
\begin{equation}
    (u(\E),u(\bE)) \equiv (u,\bar{u}).
\end{equation}

Ranging over the set of all valid CUPs in a probability theory forms a \emph{CUP-set} $\C$ -- which depends only on the underlying physical theory and the dimensions $d_X$, $d_A$ \& $d_B$. In Section \ref{sec:compatible-unitarity-pairs}, we show that CUP-sets allow us to compare and contrast fundamental aspects of different physical theories, including incompatibility. In Section \ref{sec:classical-cup-sets}, we establish that classical physics has a CUP-set on the boundary of the unit square, while in stark contrast the simplest CUP-set in quantum theory is described by a non-trivial shape in the plane (see Figure  \ref{fig:discrete-banana-ibm-Belem-iRB}). 

We explain why the shape of CUP-sets encode incompatibility and we prove (see Theorem \ref{thm:cup-bounds}) the following result:\\
\\
\textbf{Result} (Incompatibility bound on quantum CUP-set)
\textit{Any point $(u,\bar{u})$ in a quantum CUP-set lies in the band defined by}
\begin{equation}
    \frac{d_X}{d_X+1} \left(\frac{1}{d_A} + \frac{1}{d_B}\right)    \leq u + \bar{u} \leq 1
\end{equation}
\textit{where $d_X$ is the shared input system dimension, and $d_A$ \& $d_B$ are the respective output dimensions. }\\
This provides a general constraint on any quantum CUP-set.  In Section \ref{sec:quantum-cup-sets}, we relate this result to the no-cloning theorem and the impossibility of perfect hiding of quantum information under unitary evolution, to which there is no classical equivalent. Moreover, the above bounds are tight under general conditions that we discuss. Further, when $d_X=d_A=d_B$ a quantum CUP-set captures the no-hiding theorem exactly (see Theorem \ref{theorem:CUP-equivalence-relations}) which we discuss in Section \ref{sec:quantum-cup-sets}.

We next turn to the estimation of CUP-sets on quantum devices. Firstly, by directly estimating a range of CUPs using the SWAP test \cite{buhrman2001quantum}. These methods are detailed in Section \ref{sec:estimation-state-purity}. Secondly, we consider how techniques for device benchmarking \cite{wallman2015estimating,magesan2012efficient,helsen2019spectral} can be used to estimate CUPs in a SPAM-robust way, see Sections \ref{sec:estimation-through-RB} \& \ref{sec:estimation-spec-tomo}. We show that -- with some assumptions -- quantum CUP-sets can be estimated SPAM robustly on current devices (see Figure~\ref{fig:discrete-banana-ibm-Belem-iRB}).

Finally, we discuss how these methods compare, and to what degree we can infer that current devices obey the limits of quantum incompatibility.

\section{CUP-sets and incompatibility}\label{sec:compatible-unitarity-pairs}

We now construct a framework to study fundamental incompatibilities of a physical theory in a form that is sufficiently simple to allow for efficient and robust estimation. The analysis in this section focusses on quantum and classical theory, but we can extend it to any general probabilistic theory as described in Appendix~\ref{append:incompatibility-robust-measures}. 

\subsection{Unitarity of a general channel}\label{sec:unitarity-defn}
We first introduce a measure -- the unitarity -- that quantifies how noisy a channel is. This measure can also be viewed as the variance of the channel \cite{korzekwa2018coherifying}. For both quantum and classical theory, we have the notion of a physical state $x$ of a system, which may be mixed or pure~\footnote{More precisely a pure state is an extremal point in the set of all states, such as $ \ketbra{\psi}$ in quantum theory, while a mixed state is obtained from probabilistic mixtures of pure states}. For example, in classical statistical mechanics a pure state is a microstate, while a macrostate is a mixed state. The most general evolutions of states are called channels, and a channel $\E$, is simply any map that takes valid states to states. For example, the identity channel $id(x):= x$ for all states $x$. We next need a couple of additional concepts in order to define the unitarity of a channel.

Firstly, for both classical and quantum theory, we have a notion of geometry that arises for the states. In quantum theory we have the Hilbert-Schmidt inner product. For two Hermitian operators $A$ and $B$ this inner product is defined as $\<A,B\>:= \tr (AB)$, and leads to the definition of the \emph{purity} of a quantum state $\rho$ given by $\gamma := \<\rho, \rho\> = \tr (\rho^2)$. The same features exist in classical theory, and for a given probability distribution $(p_k)$ describing a classical state of a system we have its associated purity given by $\gamma(p):=\<p, p\>:= \sum_k p_k^2$. Therefore, in either classical or quantum theory, we can define the purity of a state $x$ as given by $\gamma(x):=\<x,x\>$ for the appropriate inner product. The purity provides a measure of the noisiness of a given state, and for example can be associated to the minimal collision entropy over discriminating measurements in the theory~\cite{muller2013quantum}. Moreover, this quadratic-order measure can be readily estimated for either classical or quantum theory.

Secondly, for both quantum and classical theory we have a preferred measure $d \mu(x)$ which is non-zero over the set $\partial \S$ of pure states of the theory. For quantum theory this is the Haar measure, while for finite-dimensional classical systems it is the uniform measure over the discrete pure states. 

Given this, we now define the \emph{unitarity of a channel} $\E$ as 
\begin{equation}\label{eqn:GPT-unitarity}
    u(\E) := {\rm var}(\E):= \alpha \int_{\partial \S} d\mu(x) \ \gamma( \E( x - \eta ) )
\end{equation}
where $\eta := \int_{\partial \S} \! d\mu(x)\,  x$ is the maximally mixed state under either quantum or classical theory, and where the normalizing constant $\alpha$ is chosen such that $u(id)=1$.

This unitarity measure has a range of nice properties. For example, in Lemma~\ref{lemma:gpt-unitarity-depolar-zero}, we prove that $u(\E)=0$ if and only if $\E$ is a completely depolarizing channel that acts as $\E(x) = y$ for all $x$ and some fixed $y$. Such a channel can be viewed as erasing all information in the input state of the system. Additionally, for any theory in which $\gamma(x) = \<x,x\>$ the unitarity is bounded between 0 and 1, and $u(\V) = 1$ for all isometries $\V$ (see Corollary \ref{cor:u-for-isometry}) which are transformations that perfectly preserve all information in the input state $x$. Similarly, for such theories, the unitarity is invariant under changes of basis $u(\V_1 \circ \E\circ \V_2) = u(\E)$ for any channel $\E$ and unitaries $\V_1,\V_2$ (see Lemma \ref{lemma:unitary-invariance}) \cite{cirstoiu2020robustness,wallman2015estimating}. These attributes make unitarity a natural tool for capturing the incompatibility of channels.

\subsection{Defining cloning and hiding}\label{sec:defn-cloning-and-hiding}
Given a channel $\G$ from a subsystem $X$ to subsystems $AB$ we define the \emph{marginal channels} as
\begin{align}\label{eqn:marginal-channels-no-label}
\tr_B \circ \, \G (x),  \\
\tr_A \circ \, \G (x),
\end{align}
where $\tr_A$ denotes the action of discarding the subsystem $A$, and similarly $\tr_B$ denotes discarding $B$.
With the concept of a marginal channel, we can define what it means to clone or hide a state in a theory. The ability to perfectly clone/broadcast a state can be defined as the existence of a channel $\G$ from an input system $X$ to two output systems $A$ and $B$ such that
\begin{align}
\tr_B (\G(x)) &= id(x) =x, \label{eqn:perfect-clone-e} \\
\tr_A (\G(x)) &= id(x) =x, \label{eqn:perfect-clone-e-bar}
\end{align}
for all states $x$. In other words the state is perfectly copied to the two output subsystems. Note that broadcasting is where one allows correlations between the two output systems, while cloning does not have correlations and is normally considered for pure states only. This distinction is not important here since we focus on the marginal outputs only, and henceforth we refer to the above process as cloning. The no-cloning theorem \cite{wootters1982single} can therefore be cast as a statement that --  under quantum theory -- there is no channel $\G$  such that Equations \ref{eqn:perfect-clone-e} and \ref{eqn:perfect-clone-e-bar} both hold for all states $x$.

The no-hiding theorem in its original formulation~\cite{braunstein2007quantum} says that given a quantum state $|\psi\>$ that unitarily evolves such that the output on one subsystem is a constant state -- namely a completely depolarizing channel -- then the state $|\psi\>$ can be perfectly recovered from the remaining environment subsystem. We can formulate the no-hiding theorem in terms of the above channel marginals in the following way. For a closed quantum system under unitary evolution (e.g. when $\G=\V$), if $ \tr_B(\V(x)) = y$ for some fixed state $y$, then necessarily $x$ must be completely recoverable at $\tr_A(\V(x))$. Therefore the no-hiding theorem requires that  $\tr_{B'} \circ \tr_A(\V(x)) = x$, up to final change of basis, and where the additional partial trace ($\tr_{B'}$) may be required to match the dimension of the input system.

Channel marginals can also capture a more general notion of hiding in any theory. More precisely, we say that we can perform \textit{perfect} hiding in a theory if there is a channel $\G$ from an input system $X$ to two output systems $A$ and $B$ such that for all input states $x$ we have
\begin{align}
\tr_B (\G(x)) & = \D_1(x) =y_1 \\
\tr_A (\G(x)) & = \D_2(x) =y_2,
\end{align}
where $\D_1$ and $\D_2$ are completely depolarizing channels send all states to the fixed states $y_1$ and $y_2$ respectively. In other words the marginal channels of $\G$ fully erase any information encoded in $x$. However, this is not everything. We also require that $x$ is genuinely encoded in the global correlations between $A$ and $B$. Therefore, we additionally require that $\G$ is a reversible transformation, which means there is another channel $\F$ from $A$ and $B$ to $X$ such that $ \F \circ \G(x) = id(x) = x$. This defines {perfect} hiding, but it might be possible to have partial hiding of a state, in the same way as it is possible to partially clone a quantum state. 

We next turn to quantifying how well a theory (classical, quantum or a more general theory) can both clone and hide. To completely capture hiding our framework should reproduce both quantum theory's no-hiding theorem, as well as identify perfect hiding. We do this through the above focus on marginal channels, and use the unitarity to quantify how well these local channels preserve information.

\subsection{Compatible unitarity pairs of a theory }
We now label the two marginal channels  defined in Equation (\ref{eqn:marginal-channels-no-label}) from the input system $X$ to the two output systems $A$ and $B$ as 
\begin{align}
    \E(x) = \tr_B \circ \, \G (x), \label{eqn:marginal-channels-e} \\
    \bE(x) = \tr_A \circ \, \G (x). \label{eqn:marginal-channels-ebar}
\end{align}
We name the tuple of the unitarities of these channels a \textit{compatible unitarity pair} (CUP) and use the notation:
\begin{equation}
    (u(\E),u(\bE)) \equiv (u,\bar{u}).
\end{equation}
From the previous discussion of cloning and hiding we see that the set of global channels we consider $\G$ matters. For hiding, we must consider the set of isometric channels to capture quantum theory's no-hiding theorem -- as well as the set of \emph{reversible} channels for perfect hiding \cite{nayak2006invertible}. In contrast, for cloning we are free to range over all possible channels within a theory. 

It is possible to describe a general process on a closed system in terms of reversible channels on (larger) open systems for both quantum and classical theory \cite{heunen2021bennett, bennett1973logical}. In quantum theory, global isometries suffice (captured by a Stinespring dilation \cite{watrous2018theoryvec}), however classical theory requires the use of auxiliary randomness \cite{bennett1973logical}. The \emph{isometric channels} are a proper subset of reversible channels for both classical and quantum theory and are defined as those channels $\V$ for which we have $\<\V(x), \V(y)\> = \<x,y\>$ for all states $x,y$. The smaller set of isometry channels are the traditional set considered for incompatibility in quantum theory, due to the Stinespring dilation theorem.

When $\G \in \V$, the set of isometric channels, and $\E$ \& $\bE$ are its marginal channels (as defined in  Equations~\ref{eqn:marginal-channels-e} \& \ref{eqn:marginal-channels-ebar}) then we write $\E \sim \bE$ and say that these channels are isometrically compatible. In this case, for quantum theory, the channels $\E$ \& $\bE$ are complementary to each other.

Similarly if $\G \in \R$, the set of reversible channels, then  we write $\E \sim_r \bE$. Finally, when we consider $\G$ to be the set of \emph{all} channels in a theory, we write $\E \sim_* \bE$ such that $\E$ and $\bE$ are marginals of any valid channel $\G$ from $X$ to $AB$. This notation is just to simplify definitions, and does not suggest an equivalence relation.

We now define the \textit{set} of \textit{compatible unitarity pairs} (the CUP-set) as 
\begin{align}
\C^{X \rightarrow AB} :=\{ (u(\E), u(\bE) ) \in \mathbb{R}^2 :   \E \sim \bE   \}
\end{align}
which is determined by both the structure of the particular state spaces and the admissible isometry channels in the theory. In a similar way, we define the \emph{reversible CUP-set}, $\C_r^{X \rightarrow AB}$, when $ \E \sim_r \bE $. Finally, we define the \emph{full CUP-set}, $\C_*^{X \rightarrow AB}$, for the marginals of any valid channel, when $\E \sim_* \bE$.

For the remainder of this work we shall drop the superscripts specifying the subsystems and just write $\C, \C_r$ and $\C_*$ for the CUP-sets.

Since the unitarity is bounded between $0$ and $1$, we have the following series of inclusions 
 \begin{equation}
\C \subseteq \C_r \subseteq \C_* \subseteq [0,1]^2.
\end{equation}
It turns out that CUP-sets can be defined for general probabilistic theories, and we discuss this in Appendix~\ref{append:incompatibility-robust-measures}.
Note that for any theory we have $(0,0) \in \C_*$, since we are always free to discard the input state and prepare an arbitrary constant state on the output systems (for which the unitarity vanishes). Likewise, since the identity channel is in any theory, and we are free to swap/relabel subsystems (all isometric processes) we also have that $(1,0)$ and $(0,1)$ lie in $\C$. These are common points for CUP-sets across different physical theories. 

\subsection{No-cloning and no-hiding through the CUP set}
No-cloning and no-hiding fit into this framework as follows. Firstly, if the physical theory admits perfect cloning then this implies that $(1,1) \in \C_*$ since $u(\E) = 1$  if and only if  $\E = id$ up to a final isometry \cite{cirstoiu2020robustness}.
We also note that it has been shown \cite{barnum2007generalized,chiribella2016quantum} that broadcasting is possible in a physical theory if and only if the theory has a simplex state space of perfectly distinguishable states, and so essentially only classical theory has $(1,1)$ in its CUP-sets. The no-cloning theorem can be cast compactly as a statement that -- for quantum theory, the full CUP-sets $\C_*$ (and therefore all CUP-sets) exclude the point $(1,1)$.

Secondly, from Section \ref{sec:defn-cloning-and-hiding}, the no-hiding theorem given in terms of marginal channels states that -- for quantum theory under isometric evolution --- if $\E=\D$ (a completely depolarizing channel) then necessarily the input state can be completely recovered in the other subsystem. In the case $d_X=d_A=d_B$, this implies $\bE=\V$, an isometry. We will show (see Section \ref{sec:quantum-cup-sets}) that this statement of the no-hiding theorem is captured exactly by the isometric quantum CUP-set $\C$, as for the point $(0,x)$ in $\C$ then $x=1$ only. Additionally in the case of unequal subsystems, the no-hiding theorem and the impossibility of perfect hiding implies the point $(0,0)$ must still be strictly excluded from the isometric CUP-set $\C$. We prove this also holds in Section \ref{sec:quantum-cup-sets}.

We also consider whether a theory admits perfect hiding with the addition of auxiliary randomness. This is captured by the reversible CUP-sets, $\C_r$. If the theory admits perfect hiding with auxiliary randomness then we have $\D_1 \sim_r \D_2$ for some completely depolarizing channels $\D_1$ and $\D_2$. However, as we prove in the Appendices, this statement is equivalent to the existence of the origin in the reversible CUP-set, $(0,0) \in \C_r$. We will show that the reversible CUP-sets of classical probability theory always contain $(0,0)$ whereas the quantum reversible CUP-sets only contain $(0,0)$ under certain dimensional restrictions related to mixed state purification.  

Finally, for both quantum and classical theory we examine the case with the smallest  non-trivial dimensions in detail. Since quantum physics neither admits perfect cloning, nor perfect hiding under unitary evolution, the simplest quantum CUP-sets form non-trivial subsets of the unit square $[0,1]^2$, which we discuss shortly. In contrast, both $(0,0)$ and $(1,1)$ always lie within the classical (reversible) CUP-sets.

\section{Classical CUP-sets}\label{sec:classical-cup-sets}
We now explore in detail how the classical CUP set captures the compatibility allowed in classical probability theory. We shall see that the CUP-sets of classical theory are radically different from quantum theory, and so are a simple and vivid way to contrast the two theories.

\subsection{Unitarity of classical channels}
For a classical probability distribution on a $d$ dimensional system, the pure states correspond exactly to the $d$ extremal points $\{x_i\}_{i=1}^d$ of the state space. The unitarity reduces to 
\begin{equation}
u(\E) = \frac{d}{d-1}\sum_{i=0}^{d-1}  \gamma(\E(x_i-\eta)),
\end{equation}
where $\eta = \frac{1}{d} \sum_{i=0}^{d-1} x_i$ is the maximally mixed state.

The only isometric operations with input and output systems of the same dimensions are those that permute the pure states. Recall that for any isometry we have $u=1$. Furthermore, reversible classical channels  are fully generated by the set of isometries and auxiliary classical randomness \cite{bennett1973logical,aaronson2015classification,axelsen2011reversible}, as they correspond to injective Boolean functions. This allows us to characterise CUP-sets for classical theory.

The state space of a single probabilistic classical bit is a $d=2$ system with two possible pure states $x_0 := (1,0)$ and $x_1 := (0,1)$ in $\mathbb{R}^2$. Any pure state $x_m$ encodes the bit $m \in \{0,1\}$. There are only two single-bit isometries:  the identity channel $id$ and $NOT$ operation for which  $NOT(x_0)=x_1$ \& $NOT(x_1)=x_0$. For a two-bit system $d=4$, we can define the pure states through the tensor product of the single bit pure states e.g. $x_{ab} := x_a \x x_b$ for $a,b \in\{0,1\}$.

\subsection{Cloning and hiding in classical theory}
\label{sec:cloninghide}

To clone/broadcast a classical bit in a state  $x_m := (p,1-p)$, one simply brings in an auxiliary bit in the pure state $x_0 = (1,0)$ and then performs a controlled-not (CNOT) gate, controlled on the state $x_m$ with the auxiliary bit as the target. The marginal distributions are then both given by $x_m$; the input information was perfectly copied to the marginals. In terms of channels, the protocol is simply given by
\begin{equation}
	\V(x_m) = CNOT(x_m \otimes x_0),
\end{equation}
which outputs a $2$--bit state.

Hiding of a classical (deterministic) bit involves encoding a bit $m=0$ or $m=1$ entirely in correlations so that the marginal bit states are $\eta = (1/2,1/2)$, the maximally disordered state. However, we also require that the bit $m$ is still perfectly recoverable from the total state.
This can be done as follows: we introduce a single auxiliary bit in the state $\eta$, which is viewed as an unknown key bit, and we perform a controlled-not gate on $e_m$ that is controlled on $\eta$. Equivalently we get 
\begin{equation}
	\R_{\rm{hide},1/2}(x_m) = CNOT(\eta \otimes x_m),
\end{equation}
which is a correlated $2$--bit state. It has marginals $\eta$ and also since $CNOT \circ CNOT = id$ we can perfectly recover $m$ from the joint $2$--bit state. This is the classical one-time-pad protocol for encryption  and the operation performs perfect hiding in classical theory \cite{shannon1949communication}. 

\subsection{The simplest classical CUP-set }

The most general isometries from a single bit into two bits take the form of 
$\V (x) : = \pi_{AB} (x\otimes x_0)$ where $\pi_{AB}$ is a permutation on the four basis states. There will be 6 such different isometric operations, however they produce the same 3 points on the CUP diagram as follows. As the unitarity of the marginal channels $\E$ and $\bar{\E}$ are invariant under local isometries at $A$ and $B$,  then separable operations $\pi_{AB} = \pi_A\otimes \pi_B$ will give the point $(u,\bar{u}) =(1,0)$, corresponding to $\E = id$ and $\bar{\E}=\D$.  The swap operation permuting the two systems will produce the point $(u,\bar{u} ) = (0,1)$. Finally, if the permutation corresponds to the $CNOT$ operation with control on system $A$, then on the CUP-set diagram this gives the point $(u,\bar{u})=(1,1)$, as $\E = \bar{\E} = id$. We therefore have that
\begin{equation}
\C = \{ (1,0),(0,1), \mbox{ and } (1,1)\},
\end{equation}
for the simplest non-trivial CUP-set in classical theory.

\subsection{Classical reversible CUP-set}
We now consider the CUP-set produced by the set of  reversible global operations $\R$.

The following class of operations
\begin{equation}
	\R_{p}(x) = x \x (px_0 + (1-p)x_1)
\end{equation}
introduces an auxiliary system $B$ prepared in a fixed probabilistic state. This is not isometric, but satisfies  $tr_B \circ \, \R_{p}(x) = x$, and is therefore reversible. All such channels $\R_{p}$ correspond to point $(u,\bar{u})=(1,0)$ of the CUP-sets. Generally, reversible operations are given by
\begin{equation}
	\R := \pi_{AB} \circ \R_p.
\end{equation}

Motivated by the single-bit hiding protocol in Sec.~\ref{sec:cloninghide}, we consider $\pi_{AB} = CNOT_{AB}$, the controlled-not with $A$ as the control and with $B$ as the target, that generates the following family of reversible maps 
\begin{equation}
    \R_{\rm hide, p} := CNOT_{AB} \circ \R_{p}.
\end{equation}
These are partially-hiding channels with the perfect-hiding channel occurring for $p=1/2$, corresponding to the point $(u,\bar{u}) = (0,0)$. For general $p \in [0,1]$ we have $(u,\bar{u}) = (1,p')$ with $p'=(1 - 2 p)^2$. Since we can swap output subsystems we also get $(u,\bar{u}) = (p',1)$. The remaining reversible channels are obtained from
\begin{equation}
    \R_{{\rm broad},p} := CNOT_{BA} \circ \R_{p}
\end{equation}
which gives the points $(u,\bar{u}) = (0,p')$ with $p'=(1 - 2 p)^2$ and similarly, $(u,\bar{u}) = (p',0)$ if we swap the output subsystems. We therefore have that
\begin{equation}
\C_r = \{ (t,0), (0,t), (t,1), \mbox{ and } (1,t) \mbox{ for all } t\in [0,1]\}.
\end{equation}
In other words the reversible CUP-set $\C_r$ is simply the border of the unit square $[0,1]^2$.

\subsection{Classical full CUP-set}
Finally, we consider the full 1 to 2 bit CUP-set, $\C_*$, obtained by ranging over all single-bit to two-bit systems. It can be shown (see Corollary \ref{corollary:squash-CUP-to-origin}) that if a global channel $\E$ from $X$ to $AB$ gives a point $(u,\bar{u})$ in any CUP-set and $\D$ is any global, completely depolarizing channel, then the set of convex mixtures $p \E + (1-p)\D$ give the line segment joining $(u, \bar{u})$ to $(0,0)$. This automatically implies that for classical theory we have
\begin{equation}
\C_* = [0,1]^2,
\end{equation}
since we can take convex mixtures of reversible channels with a completely depolarizing channel and the resulting line segments fill the unit square.

\section{Quantum CUP-sets}\label{sec:quantum-cup-sets}
Quantum CUP-sets are much more tightly constrained in the unit square $[0,1]^2$ than their classical counterparts, and we relate this to quantum no-go theorems. In this section we provide evidence for this statement, via tight analytical bounds on the sum of quantum CUPs. 

\begin{figure}[h]
    \includegraphics[width= 7.5cm]{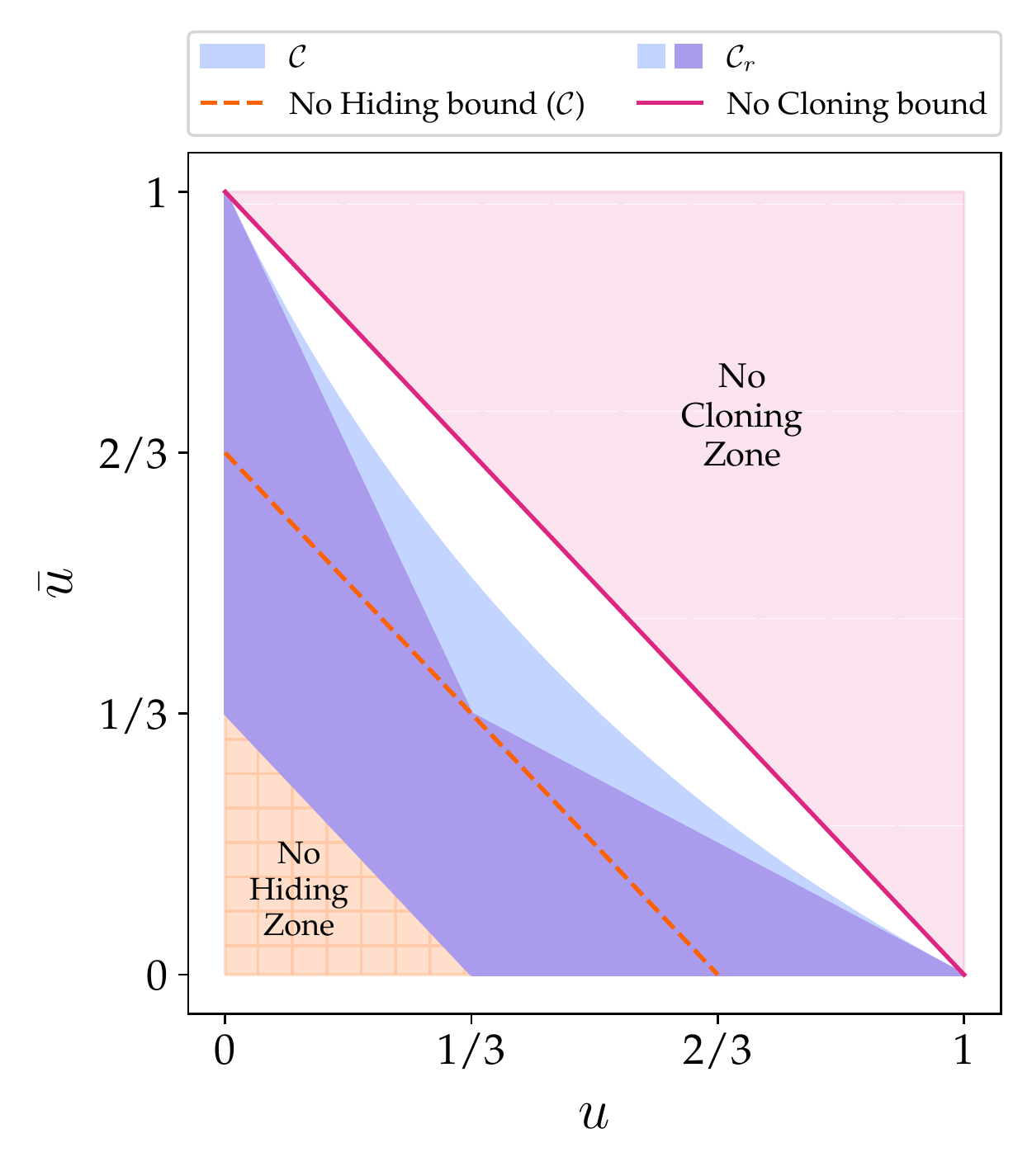}
    \caption{\textbf{Quantum CUP-sets} The simplest isometric and reversible CUP-sets under quantum theory, with their analytical bounds ($d_X = d_A = d_B =2$). The CUP-set $\C$ generated by global isometries is the central boomerang-shaped region (blue). Extending this to reversible operations $\C_r$ increases the set in the direction of $(0,0)$ to the boundary with the No-Hiding Zone (yellow). The two diagonal red lines are obtained from the general analytic upper and lower bounds for CUP-sets in quantum theory. In contrast, for classical theory we have that $\C_r$ is the border of the unit square, while $\C$ is the triple of points $(1,1), (1,0), (0,1)$.}
	\label{fig:quantum-banana}
\end{figure} 

\subsection{Unitarity of quantum channels}
Under quantum theory, from Equation (\ref{eqn:GPT-unitarity}), the unitarity of a  quantum channel $\E: \B(\H_X) \to \B(\H_Y)$ is 
\begin{equation}\label{unitarity:haar-defn}
    u(\E) := \frac{d_X}{d_X-1} \int \dd \psi  \tr[\mathcal{E}\left(\psi - \frac{\ident_X}{d_X}\right)^2]
    ,
\end{equation}
where $d_X$ is the dimension of system $X$, and where the integration is with respect to the Haar measure.

Within the context of benchmarking quantum devices, the unitarity $u$ of the average noise channel $\E$ associated with a gate-set can be estimated using Randomized Benchmarking (RB) \cite{wallman2016noise}. The unitarity of a noise channel gives additional information, beyond the average gate fidelity \cite{carignan2019bounding}. Knowing the unitarity of a channel in addition to the average gate fidelity, gives improved bounds on the diamond norm distance of a given channel to the identity \cite{wallman2015bounding}, a key figure of merit for fault-tolerant computation. Benchmarking protocols to estimate the unitarity of noise are efficient and robust against state preparation and measurement (SPAM) errors.

\subsection{Incompatibility and hiding via trade-off relations on CUP-sets}

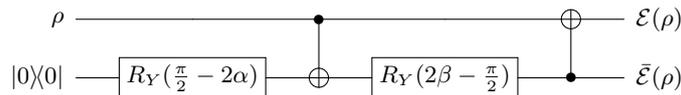
\begin{figure*}[t]
	\centering
	\begin{equation*}
		\Qcircuit @C=1.8em @R=1.2em {
			\lstick{\rho} & \qw   &  \ctrl{1}  & \qw      & \targ  & \rstick{\E(\rho)} \qw \\
			\lstick{\ketbra{0}} &  \gate{R_Y(\frac{\pi}{2}- 2 \alpha)}   & \targ & \gate{R_Y(2\beta - \frac{\pi}{2})} & \ctrl{-1} & \rstick{\bar{\E}(\rho)} \qw
		}
	\end{equation*}
	\caption{\textbf{Sufficient circuit decomposition for  2 qubit isometry $\V(\alpha,\beta)$} For $d_X=d_A=d_B=2$, the above isometry is sufficient to generate all points of the CUP-set $(u,\bar{u})$ with $0 \leq \alpha,\beta \leq \pi$ \cite{vatan2004optimal}. This follows from the general decomposition given in Figure \ref{fig:general-2-qubit-circuit}, observing that the initial two gates do not change the state of the system, and the invariance of unitarity under local unitaries.}
	\label{fig:sufficient-2-qubit-circuit}
\end{figure*}
We can now establish the following general bounds on the quantum CUP-sets that arise from isometric channels. This gives us a handle on the structure of such sets and in particular how they relate to cloning and hiding.
\begin{theorem}[General bounds on quantum CUP-set $\C$]\label{thm:cup-bounds}
    Given any input system $X$ of dimension $d_X$ and output systems $A$ and $B$ of dimensions $d_A, d_B$, with $d_X \le d_Ad_B$. The associated quantum CUP-set $\C \subseteq [0,1]^2$ is confined to the band in the $(u,\bar{u})$--plane defined by 
    \begin{equation}\label{eqn:no-hiding-bound}
        \frac{d_X}{d_X+1} \left(\frac{1}{d_A} + \frac{1}{d_B}\right) \leq u + \bar{u} \leq 1.
    \end{equation}
    This bound is tight and the quantum CUP-set $\C$ intersects the bounding lines at $(1,0), (0,1)$ and when $d_A = d_B$ it also attains the optimal hiding point $(u, \bar{u}) = (\frac{d_X}{d_A(d_X+1)}, \frac{d_X}{d_A(d_X+1)})$. 
\end{theorem}
The proof of this is provided in Appendix \ref{append:nohidingbound}.
These bounds place hard limits on the amount of quantum information that can be hidden in the correlations between systems, and also how it can be shared between local systems. The upper bound can be directly related to the no-cloning theorem, as if $u=1$ for the identity channel then necessarily $\bar{u} =0$ and the other marginal is completely depolarizing.  The perfect hiding point $(0,0)$ is always precluded from the (isometric) CUP set which is a consequence of the no-hiding theorem. Further, the upper bound is saturated  for the identity and swap-channels, while the lower bound can be saturated in the case $d_X=d_A=d_B=d$ via a $d^2$ dimensional generalization of the Controlled-NOT operation. 

In the case of equal subsystem dimensions, the quantum CUP-set is further restricted:
\begin{theorem}[No-hiding bound on quantum CUP-set $\C$]\label{theorem:CUP-equivalence-relations}
    Given any input system $X$ and output systems $A$ and $B$ all of equal dimension. The associated quantum CUP-set $\C$ is confined such that for
    \begin{equation}
        (u,\bar{u}) = (0,x) \Rightarrow x = 1.
    \end{equation}
\end{theorem}
Proof given in Appendix \ref{append:equivalence-for-qubits}. This restriction on quantum CUP-sets encapsulates both the no-hiding theorem in the following manner. The no-hiding theorem here states that if one marginal channel contains no information about the input system, then necessarily all the information can be recovered through the other marginal channel \cite{braunstein2007quantum}. Therefore when $d_X=d_A=d_B$, if $\E = \D $ then necessarily $\bE = \U$, a unitary. From Theorem \ref{theorem:CUP-equivalence-relations}, a quantum CUP-set captures this geometrically, as for the point $(0,x)$ in $\C$ then $x=1$ only. Which corresponds exactly (and only) to $\E = \D$ and $\bE = \U$, thereby capturing the no-hiding theorem.

\subsection{The simplest quantum CUP-set}
The simplest quantum CUP-set, with non-trivial bipartite outputs, is the case $d_X=d_A=d_B=2$. 
From Theorem \ref{thm:cup-bounds}, for the (isometric) CUP-set $\C$, this gives the following bounds
\begin{equation}
    \frac{2}{3} \leq u + \bar{u} \leq 1.
\end{equation}
with the optimal hiding point given by $(u,\bar{u}) = (1/3,1/3)$.

For isometries mapping single qubit to two qubits, $\V(\rho) := \U_{AB} (\rho \x \ketbra{0})$, it is sufficient to range over all unitaries $\U_{AB}$ to explore the full parameter space of $(u,\bar{u})$ for $\C$. The general form of two qubit unitaries contains at most 3 CNOTs and 3 independently parametrised single qubit rotations \cite{vatan2004optimal} (see Appendix~\ref{append:figures}). However, the two parameter isometry set with $\U_{AB} = {\U_{AB}}(\alpha,\beta)$ (for $\alpha,\beta \in \mathbb{R}$) (with circuit as in  Figure \ref{fig:sufficient-2-qubit-circuit})  generates all possible complementary channel pairs, up to local unitaries \cite{vatan2004optimal}.  As the CUP-set is invariant under local unitaries, this family suffices to fully describe it.

In Figure \ref{fig:quantum-banana} we plot this simplest CUP-set, where 3 boundary curves can be identified. The families of channels generating the boundary are of interest for structural reasons and will be key to the experimental implementation we devise.
The curved upper curve is given by a smooth interpolation between the identity channel and the SWAP channel that simply swaps the outputs on $A$ and $B$. More precisely it is given by $\U_{AB} =  SWAP^\alpha$ for $0 \leq u \leq 1$ and $0 \leq \alpha \leq 1$. The analytical relationship between $u$ \& $\bar{u}$ for the upper curve is
\begin{equation}\label{eqn:upper-boundary-of-cup-set}
    (u,\bar{u}) = (u, 3 + u - 2\sqrt{1 + 3 u}).
\end{equation} 
The analytical relationship between $u$ and $\bar{u}$ for the lower curves is linear, as shown in Figure \ref{fig:quantum-banana}. The lower right curve is given by $\U_{AB} = CNOT_{AB}^\alpha$, over the domain $\frac{1}{3} \leq u \leq 1$. While the left curve is given by $\U_{AB} = CNOT_{BA}^\alpha \circ CNOT_{AB}$ over the domain $0 \leq u \leq \frac{1}{3}$. The derivations of the boundary curves are provided in Appendix \ref{sec:analytcal-form-marginals}.

\subsection{The reversible CUP-set}
A general reversible quantum CUP-set $\C_r$ (where $d_A$, $d_B$ and $d_X$ are not necessarily equal) is given by considering the marginals of the set of globally reversible channels. For $d_X < d_A d_B $ this set will be strictly larger than the set of isometric channels. In fact, as for classical theory, we can always write any reversible channel $\R$ as the convex combination of isometries $\R = \sum_i^r p_i \V_i$ (see Corollary \ref{cor:rev-upper-bound}). Therefore we can think of the reversible CUP-set $\C_r$ as introducing auxiliary classical randomness to the isometric CUP-set $\C$. Any point $(u,\bar{u})$ in $\C_r$ will obey the same upper bound as $\C$ (see Corollary \ref{cor:rev-upper-bound}) -- adding randomness does not increase our ability to clone information. However, the existence of a non-trivial lower bound for $\C_r$, and therefore the possibility of perfect hiding, will depend on the values of $d_A$, $d_B$ and $d_X$.
\subsubsection{Perfect hiding with classical randomness}
In the case $d_X=d_A=d$ and $d_B = d^2$, the reversible CUP-set $\C_r$ contains the point $(0,0)$ and perfect hiding can be achieved. However, as the channel is both non-unitary and not isometric, it does not constitute a violation of the no-hiding theorem. The following channel illustrates the perfect hiding channel when $d=2$, and can be generalised.

Labeling the 4 Pauli operators on a single qubit as $\{P_i \} = \{\I,X,Y,Z\}$ we randomly apply an operator to the input state and record which to a classical register, such that
\begin{equation}
    \R(\rho) = \frac{1}{4} \sum_i^4 P_i \rho P_i \x \ketbra{i}.
\end{equation}
where $\{\ketbra{i}\}$ are the four computational basis states on two qubits. This channel has maximally mixed marginals, $\tr_A[\R(\rho)]= \I/2$ and $\tr_B[\R(\rho)] = \I/4$. Thus $(u,\bar{u})=(0,0)$. However there exists a quantum channel, $\R'$, such that $\R'\circ\R(\rho)=\rho$ for any state $\rho$. Physically, $\R'$ is implemented by measuring the classical register, $B$, and applying the corresponding Pauli operator to system $A$, then discarding the register. The Kraus operators, $\R'(\cdot) = \sum_i R_i' \cdot R_i'^\dag$, for this channel will be of the form:
\begin{equation}
    \{R_i'\} = \{ P_i \x \bra{i} \}
\end{equation}
It is readily seen that $\sum_i R_i'^\dag R_i' = \I$.

We can connect any channel, to a isometric CUP-set in a higher dimension through the stinespring dilation.
The following isometry
\begin{equation}
    V = \frac{1}{4} \sum_i^4 P_i \x \ket{i}_{B} \x \ket{i}_{C}
\end{equation}
gives $\R(\rho) = \tr_{C}[ V \rho V^\dag ]$ where the dimension of subsystem $C$ is $d_C = 4$. However by tracing out the $A$ subsystem, we find $\rho$ can be completely recovered in $BC$. In fact, any bipartite combination of the subsystems $A$, $B$ and $C$ defines a pair of marginal channels for the isometric CUP sets $\C$ with dimensions $(2,16)$ or $(4,8)$. The lower bound on isometric CUP-sets given in Theorem \ref{thm:cup-bounds} then guarantees that there is no arrangement of $A$, $B$ and $C$ such that both marginals are completely depolarising -- confirming that quantum information cannot be completely hidden, and can always be recovered fully in the unitary dynamics of the larger system.

\subsubsection{Boundaries of the simplest reversible CUP-set}
In the case $d_X=d_A=d_B=2$, the reversible quantum CUP-set $\C_r$ is quite similar to $\C$. It has exactly the same upper boundary but different lower boundaries which are again straight lines. We have the following analytical bounds for reversible CUPs of these dimensions
\begin{equation}
    \frac{1}{3} \leq u + \bar{u} \leq 1.
\end{equation}
Where the lower bound can be found algebraically from the general circuit decomposition of a unitary on two qubits and using the characterisation theorem of reversible channels.

The lower bounding curves are straight lines, and given by considering the marginal unitarities of the reversible channel $\R(\rho)=\U_{AB}(\rho \x \frac{\I}{2})$. The right lower surface is given by $\U_{AB} = CNOT_{AB}^\alpha$ over the domain $ \frac{1}{3}\leq u \leq 1$. The middle lower surface is $\U_{AB} = CNOT_{BA}^\alpha \circ CNOT_{AB}$ for $0 \leq u \leq \frac{1}{3}$. Finally, the left surface is given by $\U_{AB} = CNOT_{AB}^\alpha \circ CNOT_{BA} \circ CNOT_{AB}$ for $u=0$.

A similar construction to the lower boundaries of this reversible CUP set appears in the context of interleaved fidelity randomized benchmarking \cite{dubovitskii2022partial}.

\begin{figure*}[t]
	\centering
	\begin{subfigure}[t]{0.5\linewidth}
		\centering
		\begin{equation*}
			\begin{array}{c}
				\Qcircuit @C=1.2em @R=1.2em {
					\lstick{\ketbra{0}} & \qw &  \gate{H} &  \ctrl{2}   &  \gate{H} &  \rstick{\ev{Z_{i,j}}}  \qw \\
					\lstick{\ketbra{0}} &  \gate{X^{i}} &  \gate{\E} &  \qswap \qwx  & \qw  & \blacktriangleright \qw \\
					\lstick{\ketbra{0}} &  \gate{X^{j}} & \gate{\E} &  \qswap \qwx   & \qw  & \blacktriangleright \qw \\
				}
			\end{array}
		\end{equation*}
		\caption[b]{For $X^1 = X$ \& $X^0 = id$, four settings of the above circuit give an estimation of $ \gamma(\E(\I/2))$ through the relation $ \frac{1}{4}( \ev{Z_{0,0}} + \ev{Z_{0,1}} + \ev{Z_{1,0}} + \ev{Z_{1,1}} ) = \gamma(\E(\I/2))$.}
	\end{subfigure}\hfil
	\begin{subfigure}[t]{0.5\linewidth}
		\centering
		\begin{equation*}
			\begin{array}{c}
				\Qcircuit @C=1.2em @R=1.2em {
					\lstick{\ketbra{0}}  & \qw       & \qw       &  \gate{H}  & \ctrl{3} & \ctrl{4} & \gate{H} & \rstick{\ev{Z}} \qw \\
					\lstick{\ketbra{0}} & \gate{H} & \ctrl{1} &  \gate{\E}  &  \qswap &  \qw &  \qw   & \blacktriangleright \qw \\
					\lstick{\ketbra{0}}  &  \qw      & \targ  &  \qw      &  \qw    &  \qswap     & \qw      & \blacktriangleright \qw \\
					\lstick{\ketbra{0}}  & \gate{H} & \ctrl{1} &  \gate{\E}  &  \qswap &  \qw &  \qw & \blacktriangleright \qw \\
					\lstick{\ketbra{0}}  &  \qw      & \targ  &  \qw      &  \qw     &  \qswap    & \qw      & \blacktriangleright \qw \\
				}
			\end{array}
		\end{equation*}
		\caption{The above circuit gives an estimation of the Choi state purity through $\ev{Z} = \gamma(\J(\E))$.}
	\end{subfigure}
	\caption{Circuits for estimation of unitarity $u(\E)$ of single qubit channel $\E$ through state purity relations.} \label{fig:choi-state-circuits}
\end{figure*}
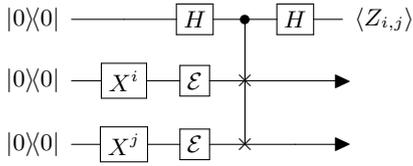
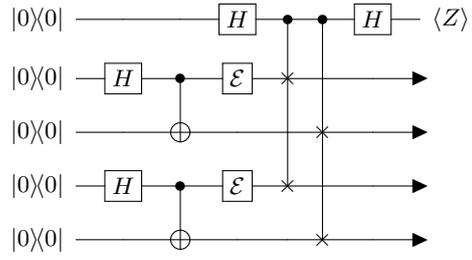

\subsection{Boundaries of the full CUP-set}
Finally, the full quantum CUP-set $\C_*$ is given by the marginal unitarities of any 1 to 2 qubit quantum channel. The upper boundary of $\C_*$ is given in Equation (\ref{eqn:upper-boundary-of-cup-set}), which it shares with both $\C_r$ and $\C$. Any point $(u,\bar{u})$ below the upper boundary is part of the full CUP-set. This is readily seen by considering a partially depolarizing channel on each output subsystem, as discussed in Section \ref{sec:depolarization}.

\section{Direct quantum CUP-set estimation through state purity measurements}\label{sec:estimation-state-purity}
Having established the relationship between CUP-sets, no-go theorems and quantum incompatibility, we now address the estimation of quantum CUP-sets. In this section, we use formulations of unitarity in terms of quantum state purities to directly estimate the individual terms.  Our simulations using IBMQ focus on two qubit systems, as this allows for the smallest non-trivial quantum CUP set. The minimal circuit decompositions  generating the boundary of the isometric quantum CUP-sets are shown in Appendix~\ref{append:figures} (Fig.~\ref{fig:lower-left-cup-set-circuit}, Fig.~\ref{fig:lower-right-cup-set-circuit}).

\subsection{Effects of noise}\label{sec:depolarization}
The methods for estimating CUP-sets we employ can be separated into two stages: (i) the preparation of the channels that generate the CUP-set, (ii) the estimation of the prepared channel's unitarity. There will be errors associated with both (i) and (ii). The errors in (i) are our primary interest as they place a limit on the device's performance at estimating CUPs. However for the direct methods we cannot easily distinguish between these errors, so refer to a noisy version, $(\cdot)_N$, of the whole process 
$(u_N,\bar{u}_N)$ for estimating $({u},{\bar{u}})$.

The simplest way to model how noise affects CUP-sets is through a  depolarizing channel given by
\begin{equation}
	\D_p := (1-p) id + p \D
\end{equation}
where $id(\rho)=\rho$ \& $\D(\rho) = \sigma$, for $\sigma$ another fixed quantum state. Given $u(id\circ \E) = u(\E \circ id)=u(\E)$ and $u(\D \circ \E) = u(\E \circ \D)=u(\D)=0$ then for any CUP-set we have
\begin{equation}
	\begin{split}
		(u_N,\bar{u}_N) &= (u(\D_{p_A} \circ \E),u(\D_{p_B} \circ \bE)), \\
		&= ((1-p_A)^2 u,(1-p_B)^2 \bar{u}).
	\end{split}
\end{equation}
Therefore by varying $p_A$ \& $p_B$ independently a CUP-set can be projected towards either axis, or towards the origin. As this allows us to reach any point in the full CUP-set ($\C_*$), we can use depolarization as a crude way to quantify how `noisy' an estimated CUP-set is ($\C$ or $\C_r$).

\subsection{Estimation through complementarity formulation}

For any quantum channel, $\E$, (with input dimension $d_X$) we can express the unitarity in terms of purities as
 \begin{equation}\label{eqn:puritys-marginal-unitarities-main-text}
 u(\E)= \frac{d_X}{d_X^2-1}\left( d_X \gamma( \tilde{\E} (\frac{\I}{d_X})) - \gamma(\E (\frac{\I}{d_X}))\right),
 \end{equation}
where $\tilde{\E}$ is any channel complementary to $\E$ \cite{cirstoiu2020robustness}. For the isometric CUP-set, $\C$, any compatible pair of channels $(\E, \bE)$ will be complementary to each other. Therefore, by estimating the two purity terms in Equation~\ref{eqn:puritys-marginal-unitarities-main-text} we get the point $({u},{\bar{u}})$.

The purity of a quantum state can be estimated through a SWAP test \cite{nielsen2002quantum}.
For two unknown quantum states, $\rho$ \& $\sigma$, the following circuit performs a SWAP test of the states
\begin{equation}\label{eqn:swap-test-circ}
    \begin{array}{c}
        \Qcircuit @C=1.2em @R=1.6em {
            \lstick{\ketbra{0}} &  \gate{H} &  \ctrl{2}   &  \gate{H}  & \rstick{\ev{Z}} \qw \\
            \lstick{\rho} &  \qw &  \qswap \qwx  &  \qw  &   \blacktriangleright \qw \\
             \lstick{\sigma} & \qw &  \qswap \qwx   &  \qw &  \blacktriangleright \qw \\
    }
    \end{array}
\end{equation}
giving $\ev{Z} = \tr[\rho \sigma]$, for the expectation value of Pauli $Z$ measured on the first qubit. The central gate is the controlled $\swap$ (or Fredkin) gate.

With $\rho=\sigma=\E(\I/2)$ or $\rho=\sigma=\bE(\I/2)$ (restricting to the case $d_X = d_A = d_B =2$), we can use the SWAP test circuit on a quantum device to get direct, albeit noisy, estimation $(\noisy{u},\noisy{\bar{u}})$ of a point $({u},{\bar{u}})$ of the CUP-set. This however requires the preparation of the maximally mixed state, which we discuss in Sec.~\ref{secV:prepartionmm}.

{
\begin{figure*}[t]
	\centering
	\begin{subfigure}[t]{0.5\linewidth}
		\centering
		\includegraphics[width=7cm]{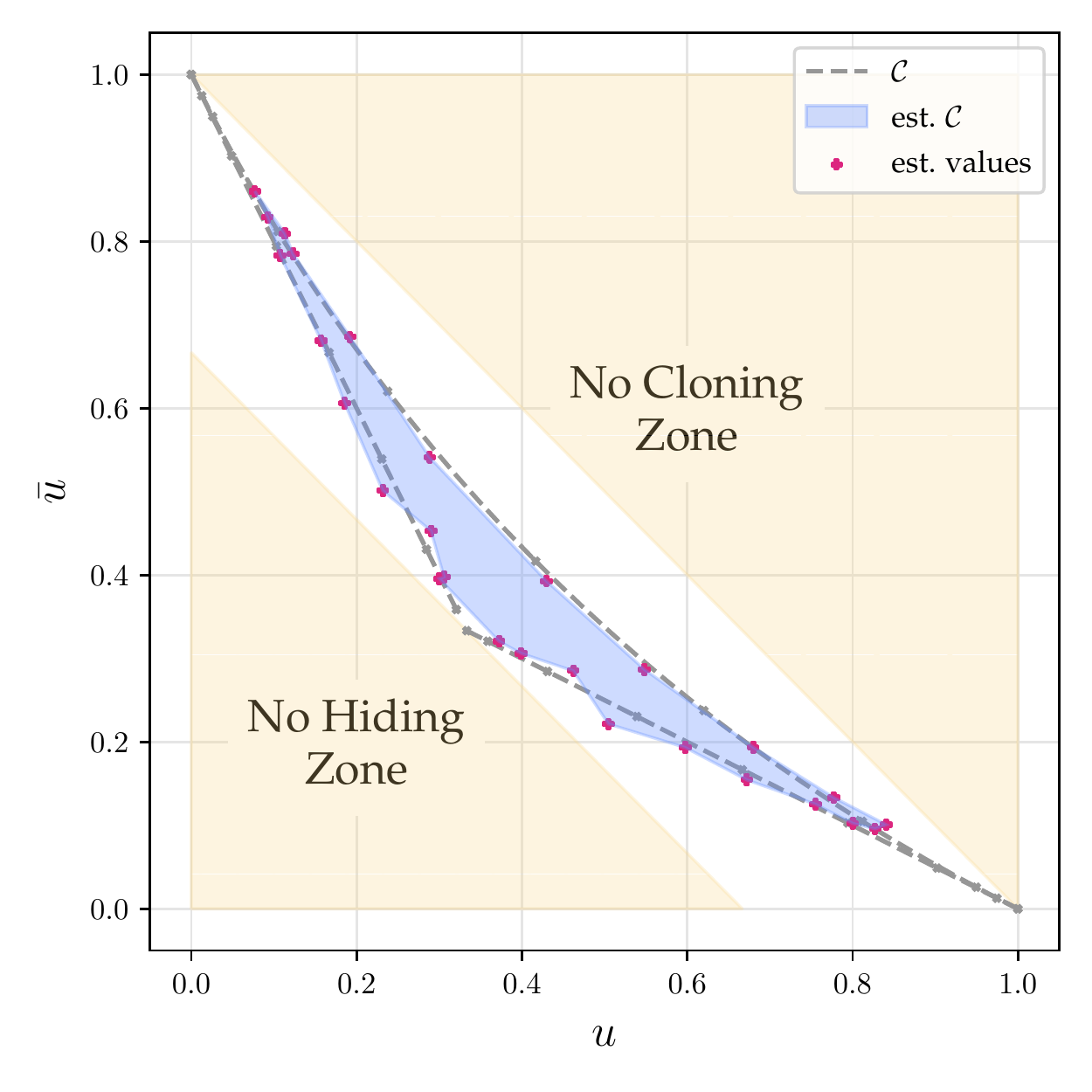}
		\caption{CUP-set ($\C$), estimated with complementary channel method.}
	\end{subfigure}\hfil
	\begin{subfigure}[t]{0.5\linewidth}
		\centering
		\includegraphics[width=7cm]{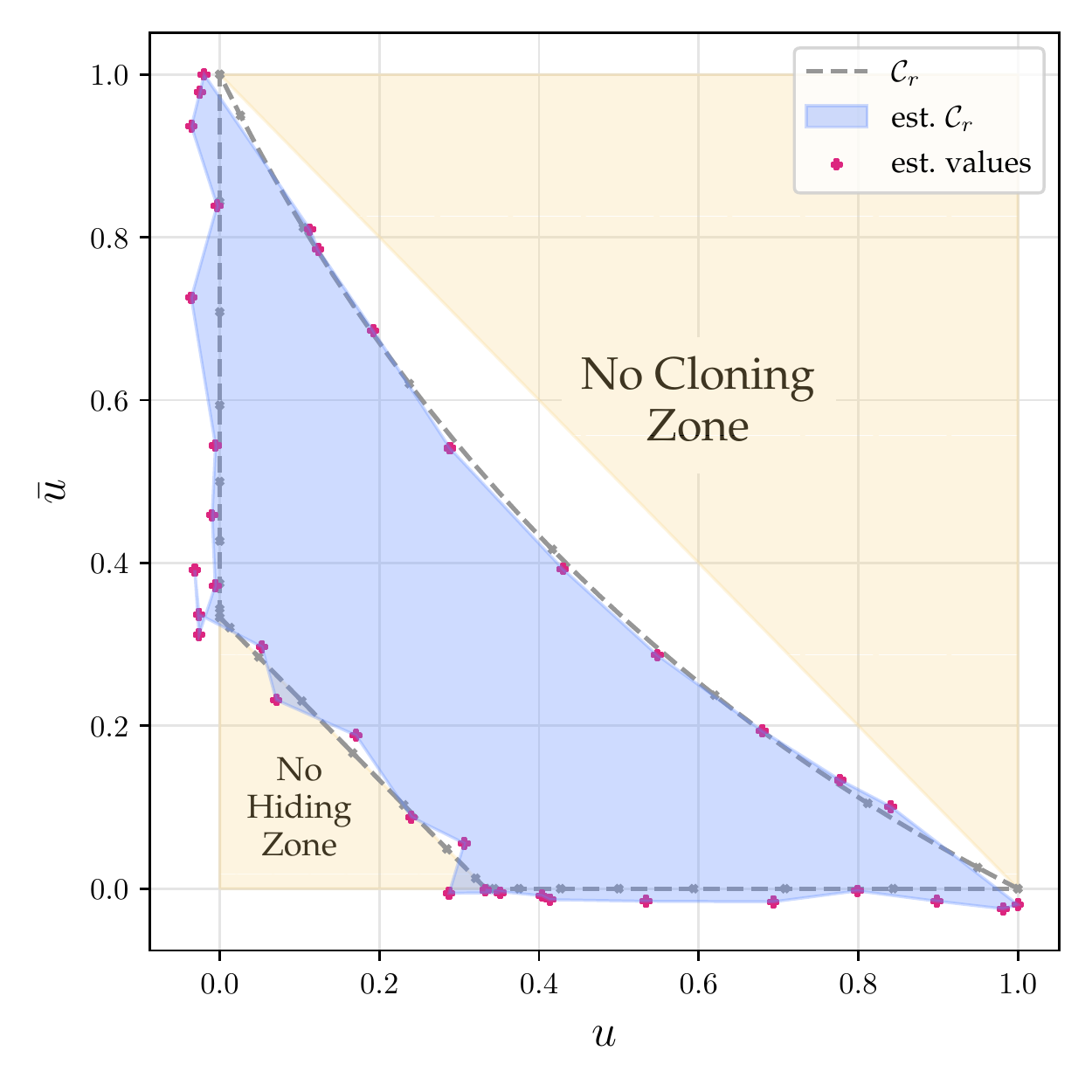}
		\caption[b]{Reversible CUP-set ($\C_r$), estimated with Choi state method.}
	\end{subfigure}
	\caption{\textbf{Direct estimation of quantum CUP-sets.} The simplest quantum CUP-sets are experimentally estimated directly through SWAP test schemes. A best fit depolarising noise model has been applied to each surface (see Table \ref{table:depolar-fit-values}).}
    \label{fig:discrete-ibm-Belem-SWAP-test}
\end{figure*}
}

\subsection{Estimation through Choi state formulation}

For the reversible CUP-set, $\C_r$, the resulting compatible pair of channels $(\E,\bE)$ are not necessarily complementary to each other. While it is straightforward to derive complementary channels for the families of channels we consider, the number of purity terms to be estimated from Equation (\ref{eqn:puritys-marginal-unitarities-main-text})  doubles compared to $\C$. Further, these new complementary channels will necessarily have a larger dimension, thereby increasing the complexity of the SWAP test. However equivalently, and perhaps more naturally, we can formulate an approach using only the channels $(\E,\bE)$ through the Choi-Jamiołkowski isomorphism.

For any quantum channel $\E$ (with input dimension $d_X$) we have
\begin{equation}\label{eqn:purity-choi-to-complementary}
    u(\E)= \frac{d_X}{d_X^2-1} ( d_X \gamma(\J(\E)) - \gamma(\E (\frac{\I}{d_X})))
\end{equation}
where $\J(\E)$ is the Choi-Jamiołkowski state of the channel $\E$ (given in Appendix \ref{append:unitarity-props}) \cite{cirstoiu2020robustness}.

Restricting to $d_X = d_A = d_B =2$, from Equation~(\ref{eqn:swap-test-circ}) we can estimate the first purity by preparing two copies of the Choi state e.g. $\rho=\sigma=\J(\E)$. For a channel with dimension $d$, the Choi state has dimension $d^2$, therefore the number of target qubits in the controlled $\swap$ for $\C_r$ is doubled compared to estimating $\C$. The second term in Equation (\ref{eqn:swap-test-circ}) can be obtained from $\rho=\sigma=\E(\I/2)$. As this process must be repeated for $u(\bE)$, estimating points on $\C_r$ will generally require twice the number of experiments of $\C$.

\subsection{Preparation of the maximally mixed state and experimental results}
\label{secV:prepartionmm}
Both of the above methods require the preparation of the maximally mixed state. With a unitary circuit, we can do this (i) statistically, by averaging the results of experiments performed on computational basis states, or (ii) by discarding information about a prepared pure state (e.g. a marginal state of a Bell state). The former method requires more experiments while the later introduces further uncertainty into the estimation.

We use (i) to estimate the isometric CUP-set $\C$ using complementarity formulation, as it requires a smaller system size. The exact circuits for the complete purity estimations are given in Figure \ref{fig:choi-state-circuits}(a). We then experimentally estimate a range of CUPs on the surface of the CUP-set $\C$  on a simulated IBM device. The results of this experiment are shown in Figure \ref{fig:discrete-ibm-Belem-SWAP-test}(a) where a partially depolarizing model has been fitted to each surface.

Then we pair (ii) with estimation through the Choi state. The exact circuits for this method are given in Figure \ref{fig:choi-state-circuits}(b). We again estimate a range of CUPs on the surface of the reversible CUP-set $\C_r$. The results of this experiment are shown in Figure \ref{fig:discrete-ibm-Belem-SWAP-test} (b) where a partially depolarizing model has been fitted to each surface.

\subsection{Discussion of direct methods}
The direct methods we have implemented have a few sources of errors. For any estimated CUP, $(\noisy{u},\noisy{\bar{u}})$, the largest error, in terms of the size of intended operation, will be on the controlled SWAP gate(s). Secondly, as the SWAP test relies upon the final measurement being taken in the correct basis, the direct methods are sensitive to even small final SPAM errors.

Examining Figure \ref{fig:discrete-ibm-Belem-SWAP-test} we observe variance in the data, even after a round of averaging over 100 experimental runs has been performed. The lack of robustness to SPAM errors, means that we cannot ascribe this variance to one source -- it may come primarily from SPAM ($\noisy{u} \approx \noisy{u}(\E)$) or it may occur in the preparation of the channel itself ($\noisy{u} \approx {u}(\noisy{\E})$). This is the main weakness with the direct methods, compared to methods we discuss in the following section.

However, we note that even after the depolarizing fit is applied, for both $\C$ and $\C_r$ the noisy estimated CUP-set is found strictly below the no-cloning upper bound, and therefore in the full CUP-set $\C_*$. 

The size of parameters needed for the depolarizing fit let us compare between the estimation of $\C$ and $\C_r$. From Table \ref{table:depolar-fit-values}, the estimated depolarization is two to three times higher for $\C_r$. As the channels required to generate $\C_r$ are very similar to $\C$, we can prescribe this increase directly to the larger overhead and complexity of the protocol for $\C_r$.

Finally, we note that  the direct methods rely on a SWAP test(s), and are therefore not efficiently scaleable in number of qubits.

\section{ SPAM robust CUP-set estimation}

\begin{figure*}[t]
	\centering
	\begin{subfigure}[t]{0.5\linewidth}
		\centering
		\includegraphics[width=7cm]{CUP-set-IBM-gateset-RB-simulation-curve-fit.pdf}
		\caption{CUP-set ($\C$)}
	\end{subfigure}\hfil
	\begin{subfigure}[t]{0.5\linewidth}
		{\centering
		\includegraphics[width=7cm]{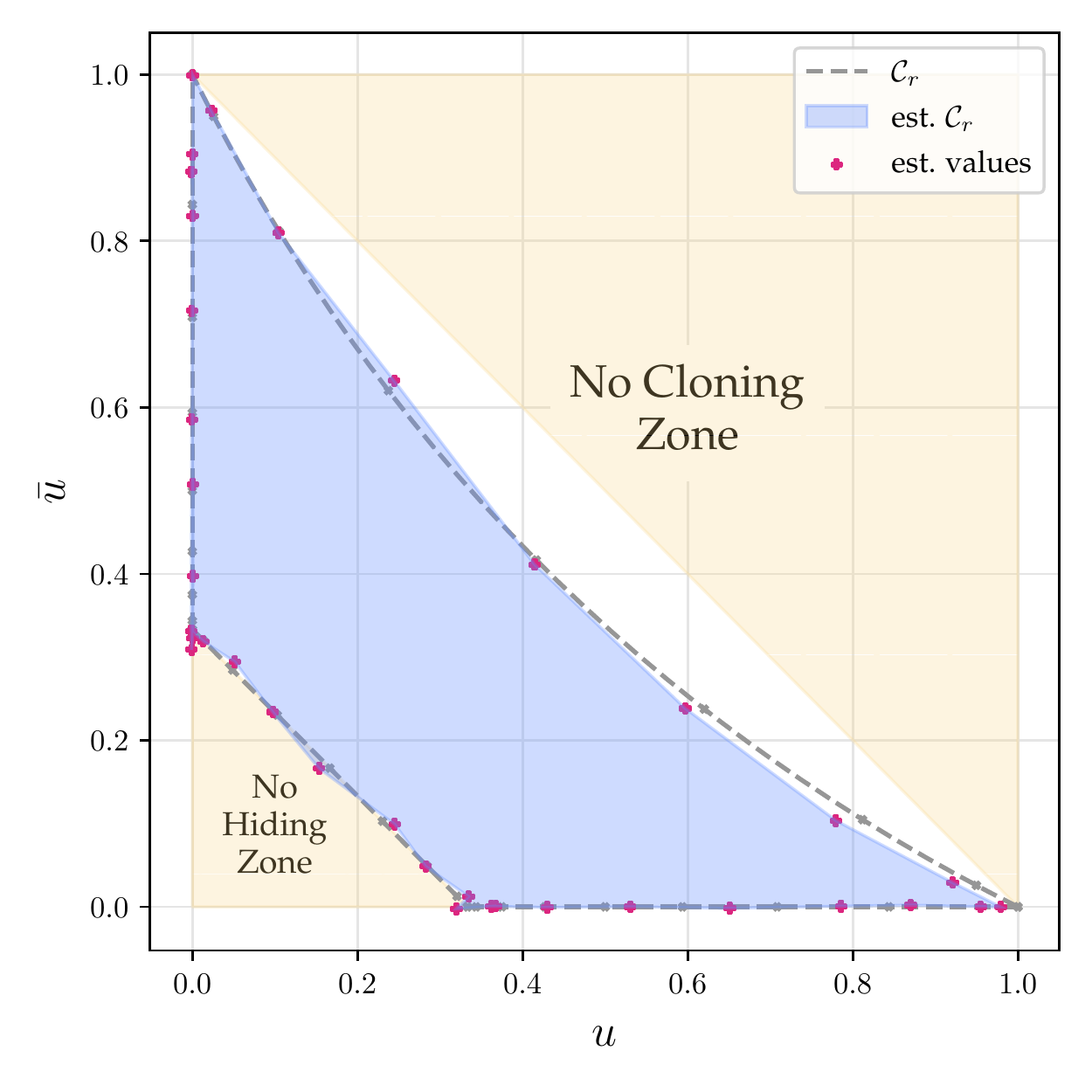}
		\caption[b]{Reversible CUP-set ($\C_r$)}}
	\end{subfigure}
	\caption{\textbf{SPAM robust estimation of quantum CUP-sets.} The simplest quantum CUP-sets are experimentally estimated through an interleaved unitarity randomized benchmarking scheme. A best fit depolarising noise model has been fitted to each surface (see Table \ref{table:depolar-fit-values}), where each surface is produced from 9 pairs of experimental values.} \label{fig:discrete-ibm-Belem-iRB}
\end{figure*}

With the direct method of the previous section, we make no distinction between errors in the implementation of the target channel, and errors in the estimation protocol including initial state preparation and final measurement SPAM errors. This severely limits the usefulness of the direct method as a measure of whether a device obeys the CUP-set's informational bounds. For example, in the extreme, we could imagine a device that implements any quantum channel perfectly but has SPAM errors such that it applies a final Hadamard transform on all qubits before measurement. With the direct SWAP test method, this would only generate the point $(0,0)$ on the CUP-set diagram. From this we might conclude the device is not acting as a closed quantum system -- when in fact, prior to measurement, it was performing perfectly.

With the above in mind, in this section we consider protocols to estimate quantum CUP-sets that are robust to SPAM errors.
However will see that the SPAM robust protocols come with a cost of much larger operational overheads, and introduce different sources of potential noise compared to the direct methods.

\subsection{Estimation through randomized benchmarking}\label{sec:estimation-through-RB}

Through randomized benchmarking (RB) \cite{wallman2015estimating} we can estimate the unitarity $u(\Lambda_C)$ of the average error channel $\Lambda_C$ induced by a computational gate-set $\{\U_C\}$ generating the Clifford group.

If we interleave a target channel of fixed dimension, $\E$, between rounds of random Clifford unitaries in the RB protocol, we can estimate the unitarity of the joint channel $u({\E} \circ \Lambda_C)$. Therefore in the limit $\Lambda_C = id$ the interleaved RB protocol returns an exact estimation of $u({\E})$. More generally, as unitarity is proportional to the Hilbert-Schmidt norm of the channel's matrix representation we also have the relation $u(\E\circ \Lambda_C) \leq u(\E)||\Lambda_C||_{\infty}$, where $||\Lambda_C||_{\infty}$ corresponds to the largest singular value of the average noisy Clifford gate-set channel. This may also be determined, for example via spectral methods as in \ref{sec:estimation-spec-tomo} to obtain more precise bounds for $u(\E)$ in the presence of noisy Clifford operations. 

Applying interleaved RB to an estimation of the CUP-set follows from the above. However, in addition, it involves an interleaved implementation of $\E$ using an ancilla initialisation, the global unitary $\U_{AB}$ and a partial trace. We require the additional assumption that we can perform mid-circuit resets, $\cal{D}(\rho):=\ketbra{0}$, and that the noisy version of these resets are incoherent -- in that none of the state $\rho$ is carried through even if $\D$ induces some larger error on the device. This allows us to include the error $\D$ in $\Lambda_C$.

Through interleaved RB we can estimate the unitarity of the following channel in a SPAM robust manner
\begin{equation}
    \noisy{\E}(\rho) = tr_B \circ \Lambda_{AB} \circ \U_{AB} \circ \Lambda_C ( \rho \x \ketbra{0} ),
\end{equation}
where $\Lambda_{AB}$ is the noise channel associated with the experimental implementation of $(\E,\bE)$, the channels generating the (isometric) CUP-set. Therefore, in the noiseless limit $\Lambda_C = \Lambda_{AB} = id$  the following protocol returns exactly $u({\E})$ in the isometric CUP-set $\C$.

\begin{protocol}[Interleaved unitarity RB for channel $\E(\rho) := \tr_B \circ \ \U_{AB}(\rho \x \ketbra{0})$.]
    \label{protocol:inter-uRB-E_A}
    \begin{enumerate}[wide, labelwidth=!, labelindent=0pt]
        \setlength{\itemsep}{2pt}
        \setlength{\parskip}{0pt}
        \setlength{\parsep}{0pt}
        \item \textbf{Prepare} the system in the state $\rho_A \x \ketbra{0}_B$.
        \item \textbf{Select} a sequence of length $k$ of random elements of the Clifford group, $\{\U_{C,i}\}$, on subsystem $A$, starting with $k=1$,  while performing a reset on subsystem $B$ after every gate. E.g. for each gate $\U_{C,i} \x \cal{D}$
        \item \textbf{Interleave} the bipartite unitary $\U_{AB}$ after every Clifford gate (such that the final gate is a Clifford gate).
        \item \textbf{Estimate} the square $(m_A)^2$, of the expectation value of an observable $M_A$ on subsystem $A$ for this particular sequence of gates.
        \item \textbf{Repeat 1, 2, 3 \& 4} for many random sequences of the same length, finding the average estimation $\mathbb{E}_\rho[(m_A)^2]$ of $(m_A)^2$.
        \item \textbf{Repeat 1, 2, 3, 4 \& 5} increasing the length of the sequence $k$ by 1.
        \item \textbf{Fit} the data $\mathbb{E}_\rho[(m_A)^2] = c_0 + c_1 s^{k-1}$ where $c_0,c_1$ are real constants, and find the estimated unitarity, $s$.
    \end{enumerate}
\end{protocol}
The above protocol gives the decay parameter that estimates $s=u(\noisy{\E})$ for the noisy channel $\noisy{\E}$ which includes the device errors from preparation of the channel $\E$, but also protocol-specific errors coming from the noisy random Cliffords.

The protocol for $\bE$ is very similar (see Appendix \ref{append:interleaveduRB}) but requires an additional $\swap$ operation after each interleaved unitary, and the resource costs associated with it. Allowing us to estimate $u(\noisy{\bE})$ for
\begin{equation}
    \noisy{\bE}(\rho) = tr_A \circ \Lambda_{AB} \circ \U_{AB} \circ \Lambda_C ( \rho \x \ketbra{0} ).
\end{equation}
Proofs showing that the above protocols indeed produce estimates of CUP sets can be found in Appendix \ref{append:interleaveduRB}. An examination of how the protocols behave under gate independent noise is given in Appendix \ref{append:general-noise-in-protocol}.

We implement these protocols on a simulated version of the \textsf{ibm belem} device, in an efficient manner (see Appendix \ref{sec:efficient-protocols}) \cite{dirkse2019efficient}. The results of the experiment for $\C$ are shown in Figure \ref{fig:discrete-ibm-Belem-iRB}(a), and for $\C_r$ in Figure \ref{fig:discrete-ibm-Belem-iRB}(b) where a depolarization model has been fitted to each surface.

\subsection{Estimation through spectral methods}\label{sec:estimation-spec-tomo}
We next discuss if spectral methods (that estimate eigenvalues of a channel) are an alternative SPAM robust path to estimate CUP-sets.  We include several results that link unitarity to quantities estimable through spectral tomography  \cite{helsen2019spectral} which may be of independent interest.

\subsubsection{Unitarity and channel eigenvalues}

Any quantum channel $\E$: $\B(\H) \to \B(\H)$  on a system $\H$ of dimension $d$ has a (Liouville) representation as a $d^2\times d^2$ matrix. Its non-unital $d^2-1 \times d^2-1$ block $T_{\E}$  has eigenvalues $\{\lambda_i(\E)\}$ that are real or come in complex conjugate pairs \cite{wolf2010inverse}. The following bound (with proof given in  Appendix \ref{append:unitarity-props}) holds for all quantum channels of fixed dimension.

\begin{lemma}\label{lemma:eigenvalue-lower-bound-rand-unit}
    For any quantum channel $\E$, of fixed dimension $d$ and unitary channels $\U_i$ and $\U_j$ we have 
    \begin{equation}
        u(\E) \geq  \sum_{k=1}^3 \frac{\abs{\lambda_k(\U_i \circ \E \circ \U_j)}^2}{d^2 - 1}.
    \end{equation}
\end{lemma}

Further, for a single qubit channel, $\E$, we can improve upon this bound:

\begin{theorem}[Variational Formulation]\label{theorem:max-eigenvalue-unitarity}
    For any single qubit quantum channel $\E$, maximising over all single qubit unitary channels $\{\U_i\}$ and $\{\U_j\}$ gives
    \begin{equation}
        u(\E) = \max_{\U_i,\U_j} \sum_{k=1}^3 \frac{\abs{\lambda_k(\U_i \circ \E \circ \U_j)}^2}{3}.
    \end{equation}
\end{theorem}
Proof given in Appendix \ref{append:unitarity-props}. The practical application of Theorem \ref{theorem:max-eigenvalue-unitarity} to the channels that generate the CUP-set is shown in Figure \ref{fig:bound-with-units}.

\subsubsection{Estimation of CUP set through spectral tomography}

Putting this together, a spectral protocol to estimate the CUP-set would require the following steps. For any point, estimate the eigenvalues of the channel $\U_i \circ \E \circ \U_j$ through spectral tomography for $N$ different randomly chosen $\U_i$ \& $\U_j$. From Lemma \ref{lemma:eigenvalue-lower-bound-rand-unit}, the set of estimated eigenvalues provide a lower bound on $u(\noisy{\E})$ where $\noisy{\E}$ is a noisy experimental implementation of $\U_i \circ \E \circ \U_j$. Repeat for $\bE$ to obtain a lower bound on $u(\noisy{\bE})$ similarly. For $N \to \infty$, and in practice for at most $N \approx 100$ (see Figure \ref{fig:bound-with-units}), from Theorem \ref{theorem:max-eigenvalue-unitarity} the estimated lower bound becomes an estimation of exactly the required unitarities. 

We performed the above sequence of spectral tomographic experiments on a simulated version of the IBMQ device, \textsf{ibm belem}. However, using a similar number of resources to the interleaved RB protocol, we were unable to extract eigenvalues accurately from the tomographic data. We suspect this was due to the finite sampling of expectation values, as for state vector simulations (without sampling) we were able to extract eigenvalues correctly. While  increasing the number of shots may therefore help, the experimental overhead would be greatly increased compared to the other SPAM robust techniques we consider.

\subsection{Discussion of SPAM robust methods}
We now briefly discuss the limitations of the interleaved randomized benchmarking technique we give for estimating CUP-sets. While the protocol is robust to SPAM errors, it relies on mid-circuit measurements to perform resets which must be incoherent (but can be noisy). Under this assumption, the decay parameter of the protocol gives a robust estimation of the unitarity of the given channel,  $s = u(\noisy{\E})$. If, as we might expect on a NISQ device, the reset allows some coherent information through, then the decay parameter can no longer be directly related to the unitarity, e.g. $s = \noisy{u}(\noisy{\E})$. For further discussion see Appendix \ref{append:general-noise-in-protocol}.

For the channels $(\noisy{\E},\noisy{\bE})$ to be close (in terms of unitarity) to the channels that generate the CUP-set $(\E,\bE)$, we need the error $\Lambda_C$ on one qubit Clifford unitaries to be small. As the error preparing $(\E,\bE)$ should be of similar size to $\Lambda_C$, then we expect that the approximately half of the depolarizing fit required in Figure \ref{fig:discrete-ibm-Belem-iRB} can be attributed to the preparation of $(\E,\bE)$.

\subsection{Comparison with direct methods}
While the SPAM robust methods require an additional assumption about the nature of resets on the device, this is a vast improvement over the direct methods of Section \ref{sec:estimation-state-purity}, where errors arising in the protocol and in the channel preparation could not be separated. The estimation of each CUP-set obtained through interleaved RB is also significantly better in terms of the required depolarizing fit than the direct methods (see Table \ref{table:depolar-fit-values}). The variance in the data points is also significantly lower for interleaved RB, even when performing an additional round of averaging for the direct methods.

Additionally, we see that the interleaved RB protocol is very good at estimating points where $u$ or $\bar{u}=0$, especially compared to the direct methods. This is likely due to the fact that, for the direct methods, these points require the estimation of two non-zero purities for any value of $u$, whereas the SPAM robust methods estimate a single decay parameter.

\section{Outlook}
\subsection{The long arm of purity}
One way of viewing the approach we have taken here, is that we are starting with the concept of purity and applying it with greater and greater abstraction in the sequence
\begin{equation}
\mbox{States} \rightarrow \mbox{Channels}\rightarrow \mbox{Theories}.
\end{equation}
Let us make this more precise.  In any general theory we can begin with an elementary notion of disorder of a state, which can be quantified via the purity $\gamma(x)$. This can now be extended to the channel level for the theory and we obtain the unitarity $u(\E)$ which is the natural generalization of purity. Indeed if we view a state $x$ as itself being a preparation \emph{channel} $1\rightarrow x$ from the trivial system to $S$ then it is readily seen that we have that $u(x) = \gamma(x)$ and the two notions coincide. For more non-trivial channels it can be shown \cite{cirstoiu2020robustness} that the unitarity coincides with the conditional purity of the Choi state of $\E$. The unitarity is a variance-based measure of the disorder of a channel from one input system to one output system. 

We next extend this further to consider how order can be shared or distributed amongst subsystems $A$ and $B$ of a theory and unitarity pairs $(u,\bar{u})$, and subsets of channels. Again, this is a generalization of the preceding concept since if $B$ is the trivial system then $(u,\bar{u}) = (u(\E), 0)$, which is just the unitarity of a channel. When applied to \emph{sets} of channels this leads to encodings of no-go results of the theory. In this sense CUP-sets are purity measures of a given physical theory within the space of all operational theories. 

\subsection{Conclusions}
We have derived a simple formulation of information disturbance and incompatibility in quantum theory, given through the set of \textit{compatible unitarity pairs} (CUPs). These pairs of compatible channels can be defined in any generalized probability theory, and they capture key limits of information transformation under the chosen theory. 

We undertook a thorough comparison between CUP-sets under quantum theory, where they are tightly bound, and classical theory where the CUP-set lies on the boundary of the unit square. We then explored the CUP-set for quantum theory in detail, including general bounds on these sets, which we related to quantum no-go theorems.

As the CUP-set encapsulates fundamental incompatibility limits of quantum theory, it may be used as a tool for benchmarking quantum devices. To this aim, we showed how the quantum CUP-set can be estimated through simple and direct purity methods, but also in a SPAM robust way through interleaved randomised benchmarking of unitarity. While estimating many points on the CUP-set may not be an efficient method of benchmarking, the extremal points (given for the qubit case by $(0,1)$, $(1,0)$ and $(1/3,1/3)$), requires just 6 experiments. The extremal points capture both the core CUP-set geometry and the unitarity-based information disturbance relation given in Theorem \ref{thm:cup-bounds}, and therefore are a natural minimal set. Future work will focus on implementing the estimation methods for CUP-set on different quantum hardware.  In particular, we may also consider randomised measurements \cite{elben2022randomized} for direct purity estimation, which would give an additional method to produce the CUP-set, with minimal implementation overhead.

Recently, many theoretical results have analysed how the effect of noise on quantum algorithms results in a computation that can be efficiently simulated classically \cite{stilck2021limitations}. This behaviour remains even for quantum advantage experiments \cite{aharonov2022polynomial}. Similarly, we have seen that noise affects the quantum CUP-set by shifting it towards regions that exhibit classical behaviour such as hiding.  An interesting future direction would be to connect these two aspects and determine if device benchmarking via CUP-sets can provide additional information to bound finite size classical simulability of quantum circuits in the presence of noise.

While we expect classical devices to perform a perfect estimation of the isometric CUP-set, the reversible classical CUP-set relies on a source of randomness to perform perfect hiding. The accuracy of the estimated CUP-set can then be directly related to the bias in the randomness. Therefore the CUP-set formulation may also be useful as a diagnostic tool in assessing the quality of a source of randomness.

Finally, in this work we primarily considered two theories, classical theory and quantum theory, however CUP-sets can be derived for more general physical theories. It would be interesting to see how the structure of CUP-sets varies between different theories.

\section*{Acknowledgments}
We thank Roger Colbeck for useful discussions on perfect hiding in  quantum theory.
MG is funded by a Royal Society Studentship.
DJ was supported by the Royal Society and also a University Academic Fellowship. We acknowledge the use of IBM Quantum services for this work. We thank Matty Hoban and Daniel Mills for useful feedback on this manuscript.


\clearpage
\onecolumngrid

\appendix

\section{General incompatibility and reversibility} \label{append:incompatibility-robust-measures}

\subsection{A general definition of unitarity}
Generalized probability theories (GPTs), provide a broad framework in which one can compare different physical theories and study their fundamental properties from an abstract, often information-theoretic viewpoint \cite{chiribella2016quantum}. Our primary aim is to capture incompatibility, and no-go theorems through measures which, at least for quantum theory, can be efficiently computed and are robust to noise. However our work can be framed in a general GPT setting, which we explore in this section.

A GPT is defined by a closed, convex set $\S$  of states, and an effects space $E$, from which the allowed measurements on $\S$ are constructed. The extremal points of $\S$ are called the pure states, and we denote this set by $\partial \S$. We shall further assume that we can embed both the state space $\S$ and effects space $E$ in a Euclidean vector space, with inner product $\<\cdot, \cdot \>$.  A measurement $\mathcal{M}$ is given by any tuple of effects $\mathcal{M}= \{ m_1, m_2, ..., m_N \}$ with $m_i \in E$ such that $\sum_k^N \ev{m_k,x} = 1$ for all states $x$ in $\S$. The probability of getting an outcome $m_k$ on a state $x$ is given by $p(m_k|x) = \<m_k,x\>$. The dimension $d$ of the state space is given by be the maximal number of completely distinguishable states $\{ x_1,x_2,...,x_d\}$ in $\S$, where a set of states is completely distinguishable if there is a measurement $\mathcal{M}^\sharp = \{ m_1, m_2, ..., m_d \}$ that unambiguously identifies which of the states was measured through its deterministic outcome. We call $\mathcal{M}^\sharp$ a \textit{sharp} measurement.   Any physical process corresponds to a channel $\E$, which is a linear map that sends any valid state $x$ in the input system to another valid state $\E(x)$ in the output system. 

We define the following function, the \emph{purity} $\gamma(x)$ of a state $x$ via
\begin{equation}\label{eqn:general-purity}
\gamma (x) := \max_{\mathcal{M}^\sharp} \sum_k \<m_k , x\>^2,
\end{equation}
where the maximization is taken over all sharp measurements $\mathcal{M}^\sharp=\{m_k\}$ in the theory \cite{chiribella2010probabilistic}. While this optimization is non-trivial it turns out that the optimal measurements are simply the measurement of the pure states in classical theory, and the rank-1 projective measurement in the eigenbasis of the state in the case of quantum theory (see Lemma \ref{lemma:quantum-purity-from-gpt}).  Additionally we can define a generalized maximally mixed state for any GPT obtained by averaging over the pure states $\partial \S$ of the theory
\begin{equation}\label{eqn:general-MM-state}
\eta := \int_{\partial \S} \!\! d\mu(x)\,  x.
\end{equation}
Together Equations (\ref{eqn:general-purity}) \& (\ref{eqn:general-MM-state}) allow for the unitarity of a channel to be calculated for theories which do not have an inner product purity $\ev{x,x}$.
More generally, the constant $\alpha$ in the definition of unitarity given in Equation (\ref{eqn:GPT-unitarity}) will generally depend on the structure of the state space $\S$ and the measure $d\mu(x)$.

The way in which the state space of subsystems relates to the state space of the global system is slightly non-trivial, and the details can be found in \cite{chiribella2016quantum,janotta2014generalized,barrett2007information,plavala2021general,hardy2001quantum}.  For composite systems we also have the notion of tracing-out or discarding of subsystems, that corresponds to the unit effect. Sometimes for clarity, we put subscripts to specify the systems involved, so that $x_A$ is a state for system $A$ and $x_{ABC}$ is a state for a tripartite system $ABC$. For a state $x_{AB}$ on a bipartite system $AB$ we assume there is channel $x_{AB} \rightarrow \tr_A [ x_{AB}] =: x_B$ that outputs a state $x_B$ on $B$ that results from discarding or ignoring system $A$. This amounts to computing the marginal of a probability distribution. We also define the identity channel as $id(x) = x$ for all $x\in \S$. Given a channel $\E$ from a subsystem $X$ to subsystems $AB$ we define the \emph{marginal channels} as
\begin{align}
\E_A(x) &:= \tr_B \circ \, \E (x) \\
\E_B (x) &:= \tr_A \circ \, \E (x).
\end{align}

\subsection{Channel compatibility in general theories}

While recent works deal with incompatibility of measurements in general theories~\cite{bluhm2022incompatibility,jenvcova2018incompatible}, one can also extend to the notion of (in)compatible channels~\cite{heinosaari2017incompatibility}.  Two (or more) channels in a theory are \emph{compatible} if they arise as marginal channels of a valid global channel within the theory. 

Given the structure of the perfect-hiding channel in classical theory, we therefore argue that to capture no-go theorems, the appropriate set of global channels to consider in a theory is the set of reversible channels.  In any theory, we say a channel $\E$ is \emph{reversible} precisely if there is a second channel $\F$ in the theory such that $ \F \circ \E (x) = x$ for all states $x\in \S$. A particular subset of reversible channels are \emph{isometry} channels $\V$, which preserve the inner product structure i.e for any pair of states $y, z$ it satisfies $\<y,z\> = \<\V(y), \V(z)\>$.  

We also note that perfect-cloning in classical theory involves an isometry channel, while perfect-hiding in classical theory involves a non-isometric, but reversible channel. Therefore if we restricted to isometric channels in a theory this would suggest that is impossible to hide a bit in classical theory, which is not true.

In light of this, we say that a theory admits perfect cloning precisely if there is a channel $\E$ from a system $X$ into a bipartite system $AB$ such that the marginal channels are both the identity channel. We also say that the theory admits perfect hiding precisely if there is a reversible channel $\R$ from $X$ into $AB$ with marginals being two completely erasing channels $\D_1$ from $X$ into $A$ and $\D_2$ from $X$ into $B$. Here a channel $\D$ is completely erasing if for all $x\in \S$ we have $\D(x) = y$ for some fixed $y$.

The set of reversible quantum channels has been fully characterised \cite{nayak2006invertible}: $\E$ is a reversible channel if and only if there is a unitary $U$ and a \emph{mixed} state $\sigma$ 
\begin{equation}
\E (\rho) = U (\rho \otimes \sigma )U^\dagger.
\end{equation}\label{rev-qchannel}
Alternatively, $\E$ is a reversible quantum channel if and only if 
\begin{equation}
\<\E(\rho), \E(\tau)\> = c \< \rho, \tau\>
\end{equation}
for all states $\rho,\tau$ and some constant $c>0$. Here the inner product is the Hilbert-Schmidt inner product given by $\<X, Y\> := \tr [ X^\dagger Y]$. The latter demonstrates that reversible channels are a natural generalization of isometry channels.

\subsection{Properties of unitarity for GPT channels}
\begin{lemma}\label{lemma:gpt-unitarity-depolar-zero}
    For any GPT in which $d\mu(x)$ is non-zero over all of $\partial S$ we have that $u(\E)=0$ if and only if $\E = \D$ a completely depolarizing channel $\D(y) = z$ for all states $y$, and $z$ fixed. 
\end{lemma}
\begin{proof}
A sum of non-negative numbers is zero if and only if each number is identically zero. Therefore we have that $ u(\E) = 0$ if and only if $\<m_k, \E(x-\eta)\> =0 $ for all $m_k$ in the optimal measurement and for all $x\in \partial S$. Since $m_k \ne 0$ for all $k$ this means that $u(\E) = 0$ if and only if $\E(x-\eta) = 0$ for all $x$, which is true if and only if $\E(x) = \E(\eta) = y$ for all $x$ and fixed $y$. 
\end{proof}

\begin{lemma}\label{lemma:unitarity-of-mix-with-depolarisation}
    For any GPT in which $d\mu(x)$ is non-zero over all of $\partial S$ we have $u(p\E + (1-p)\D) = p^2 u(\E)$ where $\E$ is any channel, and $\D$ is a completely depolarizing channel $\D(y) = z$ for all states $y$ with $z$ fixed. 
\end{lemma}
\begin{proof}
    The proof follows from the expansion of the definition of unitarity under linearity, and that $\D(x-\eta) = 0$ (from Lemma \ref{lemma:gpt-unitarity-depolar-zero}). Putting this together
    \begin{equation}
        \begin{split}
            u(p\E + (1-p)\D) &:= \alpha \int_{\partial \S} d\mu(x) \ \gamma( p\E( x - \eta ) + (1-p)\D( x - \eta ) ), \\
            &= \alpha \int_{\partial \S} d\mu(x) \ \max_{\mathcal{M}^\sharp} \sum_k \<m_k , p\E( x - \eta )\>^2, \\
            &= p^2 \alpha \int_{\partial \S} d\mu(x) \ \max_{\mathcal{M}^\sharp} \sum_k \<m_k , \E( x - \eta )\>^2 = p^2 u(\E).
        \end{split}
    \end{equation}
\end{proof}

\begin{corollary}\label{corollary:squash-CUP-to-origin}
    For any CUP-set if the global channel $\E$ from $X \to AB$ gives the point $(u,\bar{u})$ then the set of convex mixtures $p\E + (1-p)\D$ gives the point $(p^2 \ u,p^2 \ \bar{u})$, where $\D$ is a global completely depolarizing channel $\D(y) = z$ for all states $y$ with $z$ fixed.
\end{corollary}
\begin{proof}
    This follows from Lemma \ref{lemma:unitarity-of-mix-with-depolarisation} with the observation that the marginals $\tr_A \circ \D$ \& $\tr_B \circ \D$ of a completely depolarizing channel are also completely depolarising channels (to a different fixed state).
\end{proof}

\begin{lemma}\label{lemma:unitary-invariance}
    Consider a GPT in which $\gamma(x) =\<x,x\>$. Then for any isometry, $\V$, and any other channel, $\E$, we have $u(\V \circ \E) = u(\E)$.
\end{lemma}
\begin{proof}
    The proof also follows from expansion of the definition of unitarity under linearity, and that $\ev{\V(x),\V(x)} = \ev{x,x}$ for all isometries $\V$ and states $x$.
     Then
    \begin{equation}
        \begin{split}
            u(\V \circ \E) &=  \alpha \int_{\partial \S} d\mu(x) \ \gamma( \V( x - \eta ) ), \\
            &=  \alpha \int_{\partial \S} d\mu(x) \ \< \V \circ \E( x - \eta ) , \V \circ \E(x-\eta ) \>, \\
                        &=  \alpha \int_{\partial \S} d\mu(x) \ \< \E( x - \eta ) , \E(x-\eta ) \>, \\
            &=  \alpha \int_{\partial \S} d\mu(x) \ \gamma( \E(x - \eta) ) = u(\E). \\
        \end{split}
    \end{equation}
\end{proof}

\begin{corollary}\label{cor:u-for-isometry}
    Consider a GPT in which $\gamma(x) =\<x,x\>$. Then for any isometry, $\V$, we have $u(\V) = 1$.
\end{corollary}
\begin{proof}
    This follows from Lemma \ref{lemma:unitary-invariance} with $\E = id$, noting that $u(id) =  \alpha \int_{\partial \S} d\mu(x) \ \gamma( x - \eta )$.
\end{proof}

\subsection{Sharp measurements for quantum theory}
\begin{lemma}\label{lemma:quantum-purity-from-gpt}
    For any quantum state $\rho$ of dimension $d$ we have
    \begin{equation}
    \max_{\mathcal{M}^\sharp} \sum_k \<m_k , \rho\>^2 = \tr[\rho]^2,
    \end{equation}
    where the maximization is taken over all sharp measurements $\mathcal{M}^\sharp=\{m_k\}$ of dimension $d$.
\end{lemma}
\begin{proof}
    We can write any quantum state in its eigenbasis, $\rho = \sum_i^d \lambda_i \ketbra{e_i}$, such that $\tr[\rho]^2 = \sum_i^d \lambda_i^2$. As $\M^\sharp$ completely distinguishes $d$ states we have
    \begin{equation}
        \begin{split}
            \max_{\mathcal{M}^\sharp} \sum_k \<m_k , \rho\>^2 &= \max_{\mathcal{M}^\sharp} \sum_k^d ( \sum_i^d \lambda_i \tr[ m_k^\dagger \ketbra{e_i} ] )^2, \\
            &= \max_{M} \sum_k^d ( \sum_i^d \lambda_i M_{ki}  )^2 , \\
        \end{split}
    \end{equation}
    where $M_{ki}:= \tr[ m_k^\dagger \ketbra{e_i} ]$, and forms a doubly stochastic matrix where $\sum_i^d M_{ki} = \sum_k^d M_{ki} = 1$. Expanding the purity 
    \begin{equation}
        \begin{split}
            \max_{\mathcal{M}^\sharp} \sum_k \<m_k , \rho\>^2 &= \max_M \sum_{i,j,k}^d \lambda_i \lambda_j M_{ki} M_{kj}, \\
            &= \max_M \sum_{i,j}^d \lambda_i \lambda_j (M^T M)_{ij}.
        \end{split}
    \end{equation}
    As the product of any two doubly stochastic matrices is doubly stochastic, $\sum_i (M^T M)_{ij} = \sum_j (M^T M)_{ij} = 1$. We then use that the vector $\lambda = (\lambda_1,\lambda_2,...,\lambda_d)^T$ majorizes the vector  $\mu := M^T M \lambda$. Such that
    \begin{equation}
        \begin{split}
            \max_{\mathcal{M}^\sharp} \sum_k \<m_k , \rho\>^2 &= \sum_i \lambda_i \mu_i \leq  \sum_i \frac{ \lambda_i^2 + \mu_i^2 }{2} \leq \sum_i \lambda_i^2.
        \end{split}
    \end{equation}
    Where the last inequality follows the Schur-convexity of $f(x)=x^2$. Equality holds as we can choose the measurement in the eigenbasis, which attains the bound.
\end{proof}

\section{Depolarization table}\label{append:figures}
\begin{table}[h!]
    \centering
    \begin{tabular}{@{}llllllllll@{}}
    \toprule
                   \textbf{Technique}  & \# of qubits ($\C$) \ & \# of qubits ($\C_r$)  \hspace{0.3cm}  & States & Measurements & Sequences & Repetitions  & Shots  & \textbf{Total Runs} \\ \midrule
    Complementarity SWAP  \hspace{0.3cm}   &  5 & 7 &      4    &     1        &   1       &      100     & 200    &    80000       \\
    Choi SWAP     &  7 & 9 &      1    &     1        &   1       &      100     & 200    &    20000       \\
    Spectral tomography\ \ \ & 2 & 3 &  2 &     3        &   40      &      10      & 200    &    480000      \\
    Interleaved RB & 2 & 3 &     6    &     3        &   10      &      10      & 200    &    360000      \\ \bottomrule
    \end{tabular}
    \caption{Comparison of total number of experiments undertaken for the estimation of a point on the CUP distribution. For each technique, an experimental scheme was run on a simulation of a quantum device, with device noise imported from IBM Q via qiskit.}
\end{table}

\begin{table}[h!]
    \begin{tabular}{@{}lllllll@{}}
    \toprule
                    & Upper \ \  & $\C$ Left \ \ & $\C$ Right \ \ & $\C_r$ Left \ \ & $\C_r$ Middle \ \ & $\C_r$ Right \\ \midrule
    Interleaved RB \ \  &    (0.063,0.137)   & (0.032,0.125) &    (0.095,0.159)     &    (0,0.015)         &       (0.047,0.107)        &       (0,0)       \\
    SWAP test      &   (0.165,0.147)    &     (0.073,0.160)      &     (0.084,0.153)       &    (0,0.562)         &     (0.365,0.337)          &        (0.253,0)      \\ \bottomrule
    \end{tabular}
    \caption{ \textbf{Depolarization fits for noisy CUP-sets} For each experimental estimation of the quantum CUP-sets $\C$ \& $\C_r$ we fit a depolarising noise model $ ({u_N},{\bar{u}_N}) = ((1-p_A)^2 {u},(1-p_B)^2\bar{u})$ to each surface (see Section \ref{sec:depolarization}). Best fit values for $(p_A,p_B)$ are tabulated here.}\label{table:depolar-fit-values}
\end{table}

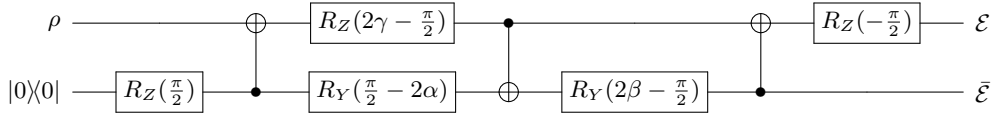
\begin{figure}[h!]
    \label{fig:genericcircuit}
    \begin{equation*}
    \Qcircuit @C=1.8em @R=1.2em {
        \lstick{\rho} & \qw    &  \targ &  \gate{R_Z(2 \gamma - \frac{\pi}{2})} &  \ctrl{1}  & \qw      & \targ      &  \gate{R_Z(-\frac{\pi}{2})}  &  \rstick{\E} \qw \\
        \lstick{\ketbra{0}} & \gate{R_Z(\frac{\pi}{2})} &  \ctrl{-1}      &  \gate{R_Y(\frac{\pi}{2}- 2 \alpha)}   & \targ & \gate{R_Y(2\beta - \frac{\pi}{2})} & \ctrl{-1} &  \qw  &  \rstick{\bar{\E}} \qw
    }
    \end{equation*}
    \caption{\textbf{Circuit decomposition for generic 2 qubit isometry $\V(\alpha,\beta,\gamma)$} For $d_X=d_A=d_B=2$, all isometries can be expressed in the above form, where $0 \leq \alpha,\beta,\gamma \leq \pi$ \cite{vatan2004optimal}. The complementary channels $\E = \tr_B \circ \V$ \&  $\bE = \tr_A \circ \V$ are shown, by ranging over $\alpha,\beta,\gamma$ we can generate the CUP-set $\C$ for 1 to 2 qubits.}
    \label{fig:general-2-qubit-circuit}
\end{figure}

\begin{figure}[h]
    \centering
    \begin{equation*}
        \Qcircuit @C=1.8em @R=1.2em {
            \lstick{\rho} & \qw   & \qw   &  \ctrl{1} & \qw  & \rstick{\E(\rho)} \qw \\
            \lstick{\ketbra{0}} &  \gate{\sqrt{X} } & \gate{R_Z(\pi(1-\alpha))}  & \targ & \gate{R_Z(-\frac{\pi}{2})} & \rstick{\bar{\E}(\rho)} \qw
        }
    \end{equation*}
    \caption{\textbf{Circuit decomposition  in IBM gateset for lower right CUP-set surface} Circuit for 2 qubit isometry $CNOT_{AB}^\alpha(\rho \x \ketbra{0})$, and complementary channels $\E$ \& $\bE$ for the lower right surface of the CUP-set are shown. The final $R_Z$ rotation is optional but aids in the estimation of CUPs through spectral techniques. }
    \label{fig:lower-right-cup-set-circuit}
\end{figure}
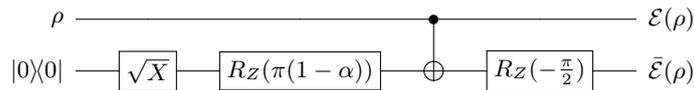

\begin{figure}[h]
    \centering

    \begin{equation*}
        \Qcircuit @C=1.8em @R=1.2em {
            \lstick{\rho}  &  \ctrl{1} & \qw   & \qw & \qw   & \targ  & \rstick{\E(\rho)} \qw \\
            \lstick{\ketbra{0}} & \targ  &  \gate{\sqrt{X} } & \gate{R_Z(\pi(1-\alpha))}   &  \gate{\sqrt{X} }  & \ctrl{-1} & \rstick{\bar{\E}(\rho)} \qw
        }
    \end{equation*}
    \caption{\textbf{Circuit decomposition in IBM gateset for lower left CUP-set surface} Circuit for  2 qubit isometry $ CNOT_{BA}^\alpha \circ CNOT_{AB}(\rho \x \ketbra{0})$, and complementary channels $\E$ \& $\bE$ for the lower left surface of the CUP-set are shown.}
    \label{fig:lower-left-cup-set-circuit}
\end{figure}
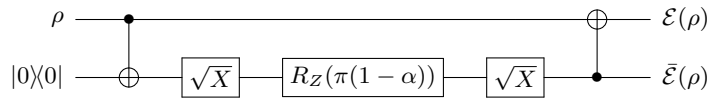

\begin{figure*}[h!]
    \centering
    \begin{subfigure}[t]{0.5\linewidth}
        \centering
        \includegraphics[width=8cm]{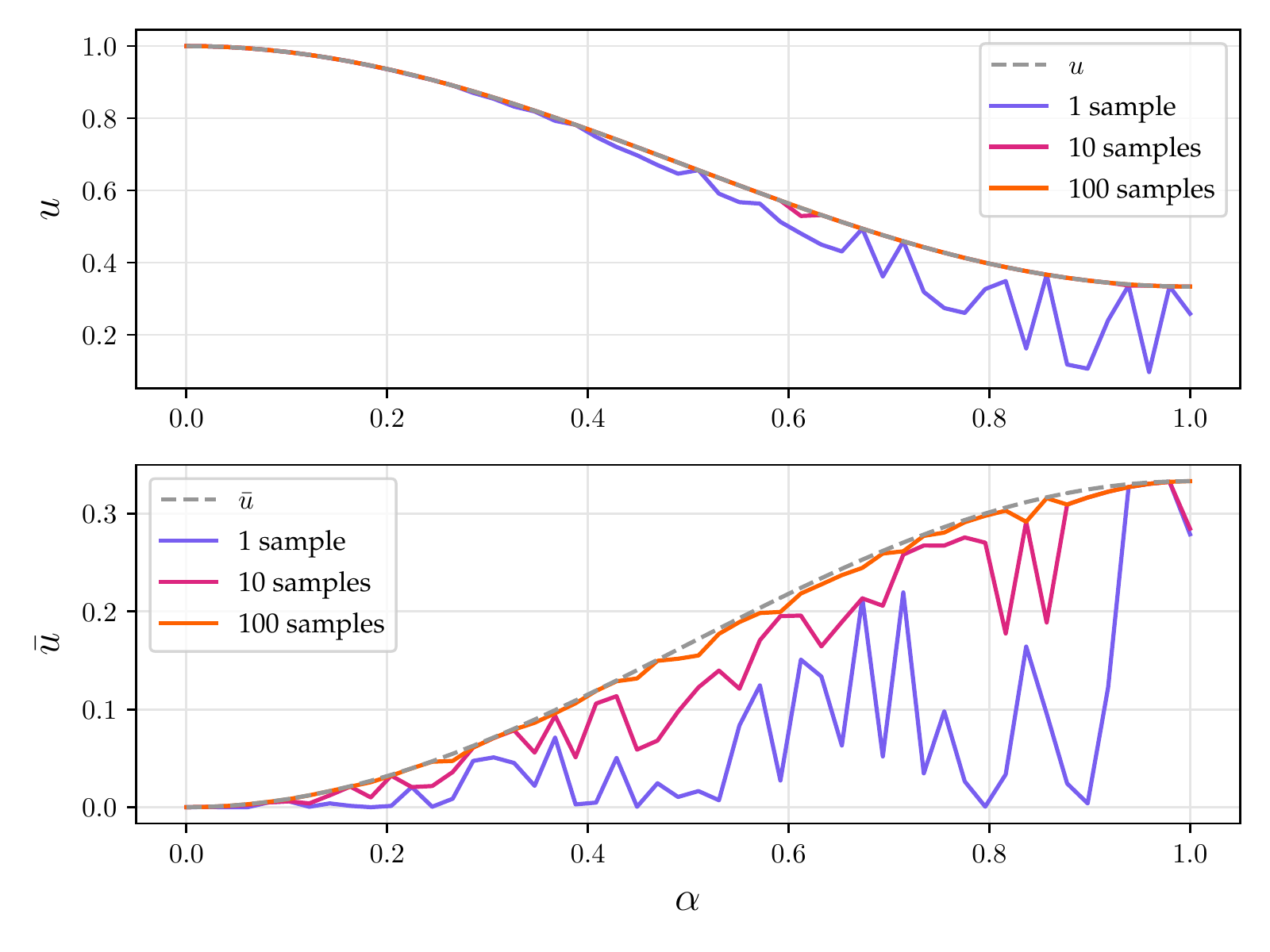}
        \caption{  $\U_{AB} = CNOT_{AB}^\alpha$}
    \end{subfigure}\hfil
    \begin{subfigure}[t]{0.5\linewidth}
        \centering
        \includegraphics[width=8cm]{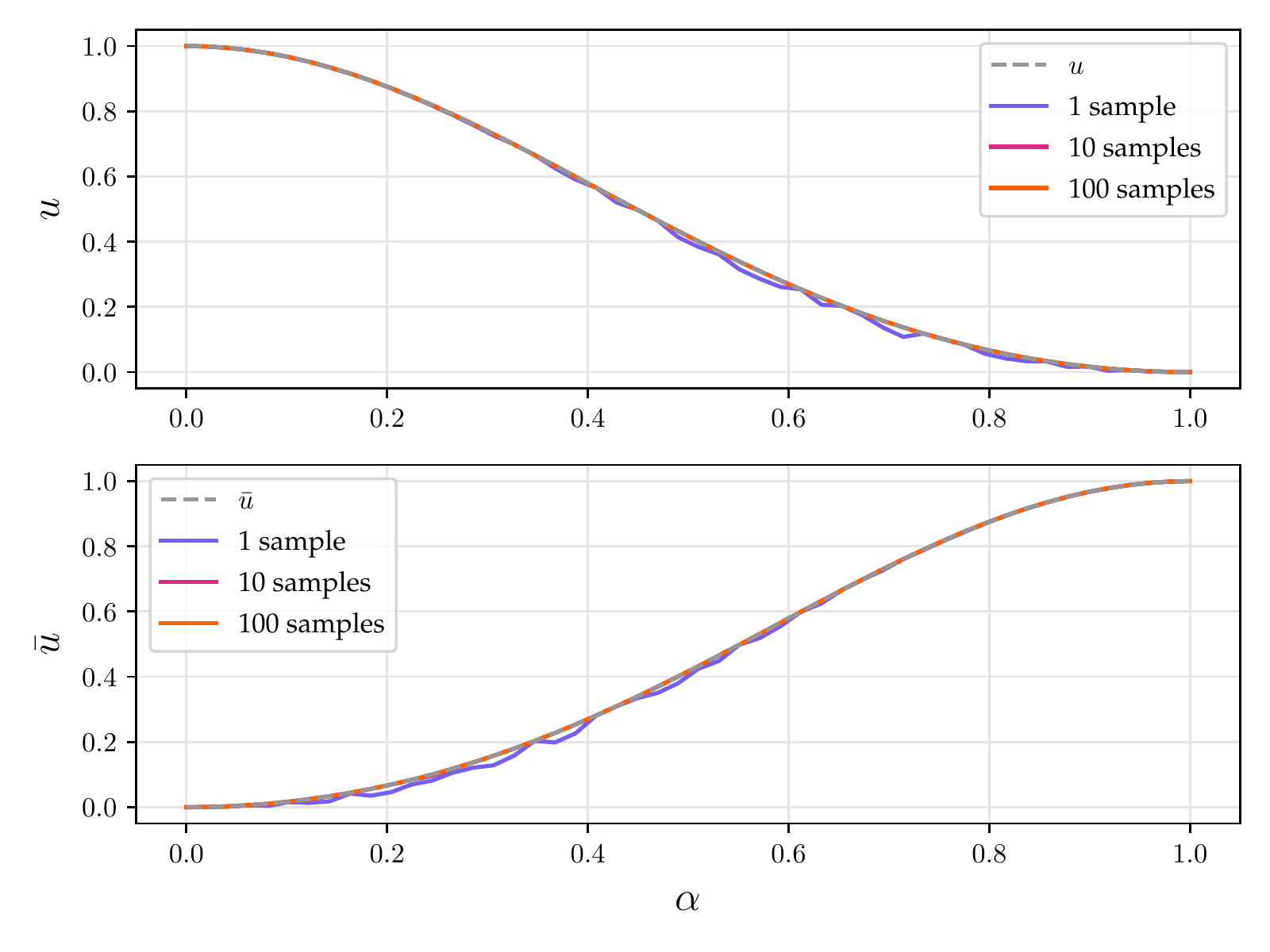}
        \caption{$\U_{AB} = \swap^\alpha$}
    \end{subfigure}
    \caption{\textbf{Bound on CUP-set through random unitaries} For two surfaces of the CUP-set $\C$ (for a range of 50 discrete values of $0 \leq \alpha \leq 1$) we can test how quickly the lower bound given in Lemma \ref{lemma:eigenvalue-lower-bound-rand-unit} converges towards the actual unitarity given in Theorem \ref{theorem:max-eigenvalue-unitarity}. Roughly, we observe, for bound within 1\% of the unitarity we need 1 random setting if $u > 2/3$, and at most 100 random settings for  lower values. } \label{fig:bound-with-units}
\end{figure*}

\begin{figure}[h!]
    \label{fig:deformation}
\centering
\includegraphics[width=7cm]{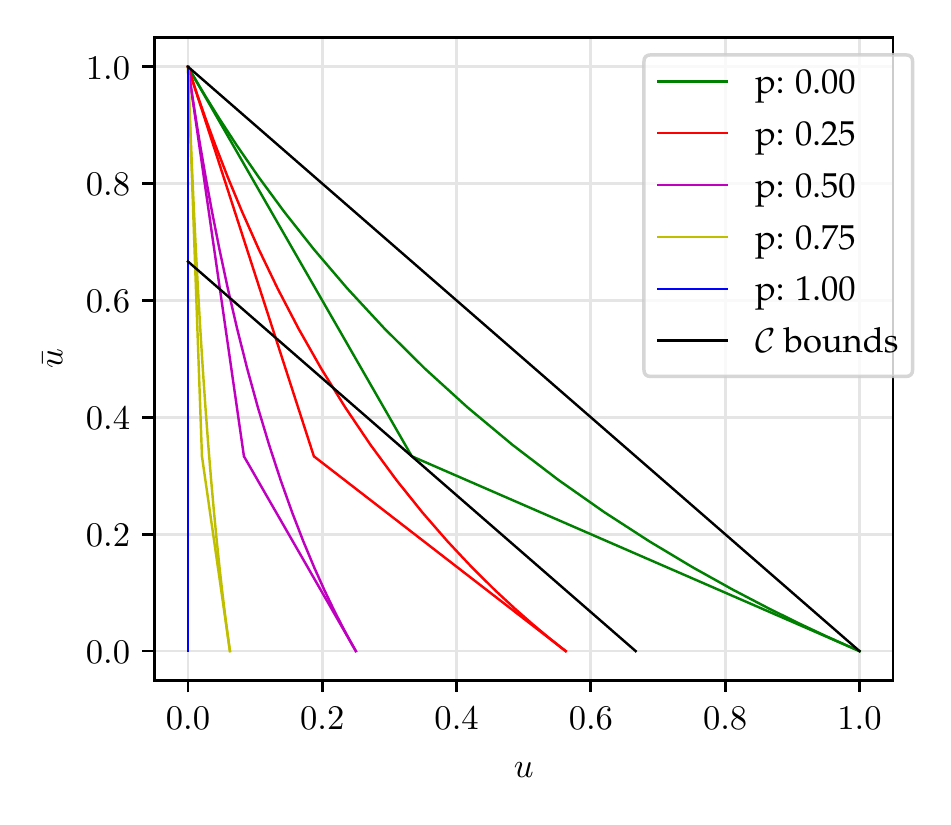}
\caption{\textbf{CUP-set deformation through depolarization} The simplest CUP-set $\C$ is shown when one output is depolarized, $(u(\D_p \circ \E),u(\bE))$  for different values of $p$.}
\label{fig:depolar-output}
\end{figure}

\section{Properties of quantum CUP-sets}
\subsection{Information disturbance bounds}\label{append:nohidingbound}

\begin{theorem}[General incompatibility bound on quantum CUP-sets]
    Given any input system $X$ of dimension $d_X$ and output systems $A$ and $B$ of dimensions $d_A, d_B$, with $d_X \le d_Ad_B$. The associated quantum CUP-set $\Q \subseteq [0,1]^2$ is confined to the band in the $(u,\bar{u})$ plane defined by 
    \begin{equation}
        \frac{d_X}{d_X+1} (\frac{1}{d_A} + \frac{1}{d_B}) \leq u + \bar{u} \leq 1.
    \end{equation}
    This bound is tight and the CUP-set $\Q$ intersects the bounding lines at $(1,0), (0,1)$ and when $d_A = d_B$ it also attains the optimal hiding point $(\frac{d_X}{d_A(d_X+1)}, \frac{d_X}{d_A(d_X+1)})$. 
\end{theorem}
\begin{proof}
It can be shown~\cite{cirstoiu2020robustness} that the unitarity of a channel can be expressed as
 \begin{equation}
 u(\E) = \frac{d_X}{d_X^2-1} (d_X\tr [ \tilde{\E} (\I/d_X)^2]  - \tr [ \E (\I/d_X)^2])
 \end{equation}
 where $ \tilde{\E}$ is any complementary channel to $\E$, which we can choose to be $\bE$. Applying the above expression to the complementary pair $(\E, \bE)$ we then have that
 \begin{equation}\label{eqn:puritys-marginal-unitarities}
 u(\E) +u(\bE) = \frac{d_X}{d_X+1} (\tr [ \E (\I/d_X)^2] + \tr [ \bE (\I/d_X)^2]).
 \end{equation}
 For the lower bound, we bound each purity term individually. For any quantum state, $\rho$, for a system of dimension $d$ we have the purity is lower bounded as $\tr[\rho^2]\geq 1/d$, and therefore
 \begin{equation}
    u(\E) +u(\bE) \geq \frac{d_X}{d_X+1} (\frac{1}{d_A} + \frac{1}{d_B}).
\end{equation}
For the upper bound, we use the following property of complementary channels. We have that $\E:= \tr_B\circ \V_{X \to AB}$ \& $\bar{\E}:= \tr_A \circ \V_{X \to AB}$ for an isometric channel $\V_{X \to AB}$. Therefore the state $\rho_{AB} = \V_{X\rightarrow AB}(\I/d_X)$ has the marginals $\rho_A = \E(\I/d_X)$ and $\rho_B = \bE(\I/d_X)$. For a general bipartite quantum state $\rho_{AB}$ we have \cite{man2014deformed} that
\begin{equation}
\gamma(\rho_A) + \gamma(\rho_B) \le 1 + \gamma(\rho_{AB}),
\end{equation}
and therefore
\begin{equation}
    u(\E) + u(\bE) \le \frac{d_X}{d_X+1} (1 +\gamma(\V_{X\rightarrow AB}(\I/d_X))).
\end{equation}
As $\V_{X\rightarrow AB}$ is an isometry
\begin{align}
    \gamma(\V_{X\rightarrow AB}(\I/d_X)) = \tr [ (V(\I/d_X) V^\dagger)^2] = \frac{1}{d_X}.
\end{align}
Substituting this into the previous inequality we obtain,
\begin{equation}
    u(\E) + u(\bE) \le 1,
\end{equation}
which completes the proof.

\end{proof}

\subsection{Equivalence relations for equal dimensions}\label{append:equivalence-for-qubits}
\begin{lemma}
    Any isometry $\V_{X \to AB}$ from an input system $X$ and joint output system $AB$ with equal dimensions $d_X=d_A=d_B=d$, defines the complementary fixed dimension channels $\E:= tr_B \circ \V_{X \to AB}$ and $\bE:= tr_A \circ \V_{X \to AB}$. The pair $\E$ and $\bE$ obey the following equivalence relations 
    \begin{align}
        \E &= \U & &\iff & \bE&= \D \nonumber \\ 
        &\Updownarrow & &~ & &\Updownarrow  \\ 
        u(\E) &= 1 & & \iff & u(\bE)&= 0 \nonumber
    \end{align}
    where $\D(\rho)=\sigma$ is a completely depolarizing channel to a fixed state $\sigma$.
\end{lemma}
\begin{proof}
    We first prove $\E = \U \iff  \bE= \D_\psi$. Note that for any quantum channel $\F$, its complementary channel $\tilde{\F}$ is unique up to an isometry on the output of $\tilde{\F}$ \cite{holevo2007complementary}. Further, we can write any isometry from $n=\log{d}$ to $2n$ qubits in the form $\V_{X \to AB}(\rho)=\U_{AB}(\rho \x \ketbra{0}^{\x n})$. Therefore it suffices to find any fixed dimension channel $\tilde{\E}$ complementary to $\E$, and apply a final unitary rotation. The isometry $\V_{X \to AB} = \U_A \x \U_B( \rho \x \ketbra{0}^{\x n} ) = \U_A(\rho) \x \U_B(\ketbra{0}^{\x n} ) $ where $\U_A$ and $\U_B$ are unitaries on the respective subsystems $A$ and $B$, gives the required form for $\E$, and we are free to set $\U_B(\ketbra{0}) = \ketbra{\psi}$ which is the exact form of $\bE$. Applying the same argument starting from $\bE$ completes the inverse direction.

    We now prove $\E = \U \iff  u(\E)=1$. For any quantum channel we have $\F=\V \iff u(\F) = 1$ for any isometry $\V$ \cite{cirstoiu2020robustness}. For fixed input and output dimensions the set of isometric channels is equivalent to the set of unitary channels, which completes the proof.
    
    Finally, we prove $\bE = \D_\psi \iff  u(\bE)=0$. For any quantum channel we have $\F=\D \iff u(\F) = 0$ for any completely depolarizing channel $\D(\rho):=\sigma$ to a fixed (potentially mixed) state $\sigma$. However given the form of $\V_{X \to AB}$ the only marginal channel that can be constructed that disregards any input state completely is given by $\bE (\rho) = \U_B(\ketbra{0}) = \ketbra{\psi}$. Therefore all completely depolarizing channels generated by $\V_{X \to AB}$ must be to pure states, $\D_\psi$, and the condition holds.
\end{proof}

\begin{corollary}
    Given any input system $X$, and output systems $A$ and $B$ of equal dimension $d_X=d_A=d_B$. The associated quantum CUP-set $\C$ is confined such that if
    \begin{equation}
        (u,\bar{u}) = (0,x) \Rightarrow x = 1.
    \end{equation}
\end{corollary}
\begin{proof}
    This follows directly from the previous lemma, and the definitions of $u$ and $\bar{u}$.
\end{proof}

\subsection{Upper bound for marginals of globally reversible channels}
\begin{lemma}\label{lemma:T-matrix-combination}
    For any convex combination of channels $\E = \sum_i^r p_i \E_i$, the respective $T_\E$ matrix has the form:
    \begin{equation}
        T_\E = \sum_i^r p_i T_{\E_i}.
    \end{equation}
\end{lemma}
\begin{proof}
    $T_\E = \sum_{j,k}^{d_X^2-1,d_Y^2-1} \braket{y_k}{\E(x_j)}  \ketbra{y_k}{x_j} = \sum_{j,k}^{d_X^2-1,d_Y^2-1} \braket{y_k}{\sum_i^r p_i \E_i(x_j)}  \ketbra{y_k}{x_j}$. However as quantum channels are linear, this is $\sum_i^r p_i \sum_{j,k}^{d_X^2-1,d_Y^2-1} \braket{y_k}{ \E_i(x_j)}  \ketbra{y_k}{x_j} = \sum_i^r p_i T_{\E_i}$.
\end{proof}

\begin{lemma}
    The unitarity $u(\E)$ is a convex function of any quantum channel $\E$.
\end{lemma}
\begin{proof}
    From Lemma \ref{lemma:T-matrix-combination}, for any convex combination of channels $\E = \sum_i^r p_i \E_i$ the corresponding $T_\E$ matrix is the convex combination of each individual term, $T_\E = \sum_i^r p_i T_{\E_i}$. 
    All norms are convex non-negative functions, including the $l_2$ norm, $||\cdot||$ \cite{horn2012matrix}. Further,  if $f(x)$ is convex and non-negative function of $x$ then $f(x)^2$ is also convex. Therefore $||T_\E||^2$ is a convex function of $T_\E$.
    Putting this together with the appropriate dimension factor we have
    \begin{equation}
        u(\E) = \alpha ||\sum_i^r p_i T_{\E_i}||^2 \leq \sum_i^r p_i \alpha ||T_{\E_i}||^2 = \sum_i^r p_i u(\E_i)
    \end{equation}
    which completes the proof.
\end{proof}

\begin{corollary}\label{cor:rev-upper-bound}
    For a reversible channel $\R$ to a bipartite system $AB$, we define the marginals $\E_\R := \tr_B \circ \R$ and $\bE_\R := \tr_A \circ \R$. The sum of the marginal unitarities obey the following bound:
    \begin{equation}
        u(\E_\R) + u(\bE_\R) \leq 1.
    \end{equation}
\end{corollary}
\begin{proof}
    We can always write a reversible channel $\R$ as the convex combination of isometries $\V_i$ as 
    \begin{equation}
        \R(\rho)=\U_{AB}(\rho \x \sigma) = \U_{AB}(\rho \x \sum_i^r p_i \psi_i) = \sum_i^r p_i \U_{AB}(\rho \x \psi_i) = \sum_i^r p_i \V_i
    \end{equation}
    for some pure states $\psi_i = \ketbra{\psi_i}$.
    Therefore we can write the marginal channel $\E_\R$ as
    \begin{equation}
        \begin{split}
            \E_\R &= \tr_B \circ \R = \sum_i^r p_i \tr_B \circ \V_i = \sum_i^r p_i \E_i 
        \end{split}
    \end{equation}
    and similarly for $\bE_\R$. As the unitarity is convex
    \begin{equation}
        u(\E_\R) = u(\sum_i^r p_i \E_i) \leq \sum_i^r p_i u(\E_i)
    \end{equation}
    and similarly for $\bE_\R$. As $u(\E_i) + u(\bE_i) \leq 1$ for all $i$, we have
    \begin{equation}
        u(\E_\R) + u(\bE_\R) \leq \sum_i^r p_i (u(\E_i) + u(\bE_i)) \leq \sum_i^r p_i = 1.
    \end{equation}
    Which completes the proof. 
\end{proof}

\section{Properties of unitarity}\label{append:unitarity-props}
\subsection{Properties of the eigenvalues of $T$ and a bound on unitarity}\label{append:eigenvalue-props}

\begin{lemma}
    For any quantum channel $\E$, of fixed dimension $d$ and unitary channels $\U_i$ and $\U_j$ of the same dimension, we have
    \begin{equation}
        u(\E) \geq  \sum_k^3 \frac{\abs{\lambda_k(\U_i \circ \E \circ \U_j)^2}}{d^2 - 1}
    \end{equation}
    where for any quantum channel $\F$, $\{\lambda_k(\F)\}$ are the eigenvalues of the associated matrix $T_\F$. 
\end{lemma}
\begin{proof}
    For any quantum channel $\F$, $T_\F$ is real matrix \cite{girling2021estimation}. From Weyl's Majorant Theorem \cite{bhatia2013matrix}, $ \sum_k^{d^2-1} \abs{\lambda_k(\F)^2} \leq \sum_k^{d^2-1} \sigma_k(\F)^2$ , where $\lambda_k(\F)$ denote eigenvalues and $\sigma_k(\F)$ singular values of $T_\F$. With $\F = \U_i \circ \E \circ \U_j$ for unitary channels $\U_i$ \& $\U_j$ it follows that
    \begin{equation}
        \sum_k^{d^2-1} \abs{\lambda_k(\U_i \circ \E \circ \U_j)^2} \leq \sum_k^{d^2-1} \sigma_k(\U_i \circ \E \circ \U_j)^2.
    \end{equation}
    Further more 
    \begin{equation}
        \sum_k^{d^2-1} \sigma_k(\U_i \circ \E \circ \U_j)^2 = \sum_k^{d^2-1} \sigma_k(\E)^2 = (d^2 -1) \ u(\E)
    \end{equation}
    from the invariance of singular values under unitary rotations, which completes the proof.
\end{proof}

\begin{lemma}
    For any single qubit quantum channel $\E$, over all single qubit unitary channels $\{\U_i\}$ and $\{\U_j\}$ we have
    \begin{equation}
        u(\E) = \max_{\U_i,\U_j} \sum_k^3 \frac{\abs{\lambda_k(\U_i \circ \E \circ \U_j)^2}}{3}.
    \end{equation}
    where  for any quantum channel $\F$, $\{\lambda_k(\F)\}$ are the eigenvalues of the associated matrix $T_\F$.
\end{lemma}
\begin{proof}
    From the previous lemma we have $\sum_k^3 \abs{\lambda_k(\U_i \circ \E \circ \U_j)^2} \leq \sum_k^3 \sigma_k({\U_i \circ \E \circ \U_j})^2$ for all $\U_i$ \& $\U_j$. However for a single qubit channel, $\E$, we can always find \cite{bengtsson2017geometry} two specific unitaries  ${\U_{AB}}_1$ \& ${\U_{AB}}_2$  such that $T_{{\U_{AB}}_1 \circ \E \circ {\U_{AB}}_2}$ is a diagonal matrix, and therefore the eigenvalues and singular values are equal $\{\lambda_i({{\U_{AB}}_1 \circ \E \circ {\U_{AB}}_2}) \} = \{ \sigma_i({{\U_{AB}}_1 \circ \E \circ {\U_{AB}}_2}) \}$ and
    \begin{equation}\label{eqn:singular-vals-inequal}
        \sum_k^3 \abs{\lambda_k({{\U_{AB}}_1 \circ \E \circ {\U_{AB}}_2})^2} = \sum_k^3 \sigma_k({{\U_{AB}}_1 \circ \E \circ {\U_{AB}}_2})^2,
    \end{equation}
    saturating the inequality. Therefore over the complete set of single qubit unitary channels $\{\U_i\}$ and $\{\U_j\}$, from the above two equations
    \begin{equation}
        \max_{\U_i,\U_j} \sum_k^3 \abs{\lambda_k({\U_i \circ \E \circ \U_j})^2} = \sum_k^3 \sigma_k({{\U_{AB}}_1 \circ \E \circ {\U_{AB}}_2})^2
    \end{equation}
    which completes the proof.
\end{proof}

\subsection{Choi-Jamiołkowski isomorphism}
For a quantum channel $\E$ with input dimension $d_X$  the Choi-Jamiołkowski state is given by
\begin{equation}
    \J(\E) := \E \x id( \psi )
\end{equation}
where $\psi = \ketbra{\psi}$ with $ \ket{\psi} := \frac{1}{\sqrt{d_X}} \sum_i^{d_X} \ket{i} \x \ket{i}$, a generalized Bell state \cite{nielsen2002quantum}.

\section{Analytical form of marginal unitarities for surfaces}\label{sec:analytcal-form-marginals}
\subsection{Analytical form for marginal unitarities of $SWAP^\alpha$ isometry}\label{sec:swapalpha}

\begin{lemma}\label{lemma:swap-marginal-unitarities}
    For the isometry $ \V_{\alpha}(\rho):= SWAP^\alpha ( \rho \x \ketbra{0} )$ where $0 \leq \alpha \leq 1$, we define the marginals $\E_{\alpha} (\rho) := \tr_B[\V(\rho)_{\alpha}]$ and $\bE_{\alpha} (\rho) := \tr_A[\V(\rho)_{\alpha}]$. The unitarities of each marginal are
    \begin{equation}
        u(\E_{\alpha}) = \frac{(1-s)(3 -s )}{3}
    \end{equation}
    and
    \begin{equation}
        u(\bE_{\alpha}) = 1 - \frac{(1-s)(3+s)}{3}
    \end{equation}
    respectively, where $s=\sin^2(\frac{\pi \alpha}{2})$. 
\end{lemma}
If we consider the sum of the marginals from Lemma \ref{lemma:swap-marginal-unitarities} we have
\begin{equation}
    u(\E_{\alpha}) + u(\bE_{\alpha}) = 1- \frac{2s(1-s)}{3}
\end{equation}
with $0 \leq s \leq 1$ and produce a tighter bound on the marginals, namely for any isometry with $d_X=d_A=d_B=2$, for a given $u(\E)$ we have
\begin{equation}
    u(\bE) \leq 3 + u(\E) - 2 \sqrt{1 + 3 u(\E)}.
\end{equation}

\begin{proof}{(of Lemma \ref{lemma:swap-marginal-unitarities})}
First we must obtain a useful analytical form for $SWAP^\alpha$. 
As $SWAP$ is a unitary channel to derive the analytical form it is sufficient to find the unitary matrix $U$ that transforms the two qubit pure state $\ket{\psi}\x\ket{\phi}$ such that
\begin{equation}
    U\ket{\psi}\x\ket{\phi} = \ket{\phi}\x\ket{\psi}.
\end{equation}
From this definition we can write
\begin{equation}
    \begin{split}
        U &= \ketbra{00}{00} + \ketbra{10}{01} + \ketbra{01}{10} + \ketbra{11}{11}, \\
    &= \frac{1}{2}(\I^{\x 2} + X^{\x 2} + Y^{\x 2} + Z^{\x 2}),
    \end{split}
\end{equation}
where $\{ \I, X, Y, Z \}$ are the Pauli matrices on 1 qubit.
Defining the Bell states as $\ket{\Phi_\pm}:=\frac{1}{\sqrt{2}}(\ket{00} \pm \ket{11})$ and $\ket{\Psi_\pm}:=\frac{1}{\sqrt{2}}(\ket{01} \pm \ket{10})$, we can diagonalise this unitary as
\begin{equation}
    \begin{split}
        U &= \ketbra{\Phi_+} + \ketbra{\Phi_-} + \ketbra{\Psi_+} - \ketbra{\Psi_-}, \\
        &= \ketbra{\Phi_+} + \ketbra{\Phi_-} + \ketbra{\Psi_+} + e^{i \pi} \ketbra{\Psi_-}.
    \end{split}
\end{equation}
As $SWAP^\alpha(\rho) = U^\alpha \rho (U^\alpha)^\dag$, up to a global phase we can find \cite{fan2005optimal}
\begin{equation}
    U^\alpha = \ketbra{\Phi_+} + \ketbra{\Phi_-} + \ketbra{\Psi_+} + e^{i \pi \alpha} \ketbra{\Psi_-},
\end{equation}
and through careful expansion
\begin{equation}
    \begin{split}
        U^\alpha &= \ketbra{00} + \ketbra{11} + \frac{1}{2}(1+e^{i \pi \alpha})(\ketbra{01}{01} + \ketbra{10}{10}) + \frac{1}{2}(1-e^{i \pi \alpha})(\ketbra{01}{10} + \ketbra{10}{01}), \\
        &= \frac{1}{2}( \I^{\x 2} +  Z^{\x 2} ) + \frac{1}{4}(1+e^{i \pi \alpha})(\I^{\x 2} -  Z^{\x 2}) + \frac{1}{4}(1-e^{i \pi \alpha})(X^{\x 2} + Y^{\x 2}), \\
        &= \frac{1}{2}(1+e^{i \pi \alpha})\I^{\x 2} + \frac{1}{4}(1-e^{i \pi \alpha})(\I^{\x 2} + X^{\x 2} + Y^{\x 2} + Z^{\x 2}), \\
        &= \frac{1}{2}(1+e^{i \pi \alpha})\I^{\x 2} + \frac{1}{2}(1-e^{i \pi \alpha})U.
    \end{split}
\end{equation}
If we now expand the isometry definition we have
\begin{equation}
    \begin{split}
        \V(\rho)_{\alpha} &= U^\alpha \rho \x \ketbra{0} (U^\alpha)^\dag, \\
        &= (\frac{1}{2}(1+e^{i \pi \alpha})\I^{\x 2} + \frac{1}{2}(1-e^{i \pi \alpha})U) (\rho \x \ketbra{0} ) (\frac{1}{2}(1+e^{-i \pi \alpha})\I^{\x 2} + \frac{1}{2}(1-e^{-i \pi \alpha})U^\dag), \\
        &= \cos(\frac{\pi \alpha}{2})^2 \I^{\x 2}(\rho \x \ketbra{0} )\I^{\x 2} + \sin(\frac{\pi \alpha}{2})^2 U(\rho \x \ketbra{0} ) U^\dag  \\
        & \hspace{2cm} + \frac{i}{2}\sin(\pi \alpha)\I^{\x 2} (\rho \x \ketbra{0} ) U^\dag -  \frac{i}{2}\sin(\pi \alpha) U (\rho \x \ketbra{0} ) \I^{\x 2}, \\
        &= \cos(\frac{\pi \alpha}{2})^2 \rho \x \ketbra{0}  + \sin(\frac{\pi \alpha}{2})^2 \ketbra{0} \x \rho   \\
        & \hspace{2cm} + \frac{i}{2}\sin(\pi \alpha)\I^{\x 2} (\rho \x \ketbra{0} ) U^\dag -  \frac{i}{2}\sin(\pi \alpha) U (\rho \x \ketbra{0} ) \I^{\x 2}. \\
    \end{split}
\end{equation}
From this point it is relatively straightforward to show that the unital block $T$ for the 1 qubit channels $\E_{\alpha}$ \& $\bE_{\alpha}$ will be
\begin{equation}
    T_{\E,\alpha} =
    \bordermatrix{~ 
        & \vket{X/\sqrt{2}} & \vket{Y/\sqrt{2}} & \vket{Z/\sqrt{2}} \cr
        \vbra{X/\sqrt{2}} &   \cos(\frac{\pi \alpha}{2})^2 & \frac{1}{2}\sin(\pi \alpha) & 0   \cr
        \vbra{Y/\sqrt{2}}  & -\frac{1}{2}\sin(\pi \alpha) &  \cos(\frac{\pi \alpha}{2})^2 & 0 \cr
        \vbra{Z/\sqrt{2}}  & 0 & 0 &  \cos(\frac{\pi \alpha}{2})^2 \cr
        },
\end{equation}
and
\begin{equation}
    T_{\bE,\alpha} =
    \bordermatrix{~ 
        & \vket{X/\sqrt{2}} & \vket{Y/\sqrt{2}} & \vket{Z/\sqrt{2}} \cr
        \vbra{X/\sqrt{2}} &   \sin(\frac{\pi \alpha}{2})^2 & -\frac{1}{2}\sin(\pi \alpha) & 0   \cr
        \vbra{Y/\sqrt{2}}  & \frac{1}{2}\sin(\pi \alpha) &  \sin(\frac{\pi \alpha}{2})^2 & 0 \cr
        \vbra{Z/\sqrt{2}}  & 0 & 0 &  \sin(\frac{\pi \alpha}{2})^2 \cr
        },
\end{equation}
respectively.
Therefore the unitarity of $\E_{\alpha}$ is given by
\begin{equation}
    u(\E_{\alpha}) = \frac{1}{3} \tr[T_{\E,\alpha}^\dag T_{\E,\alpha}] = \cos(\frac{\pi \alpha}{2})^4 + \frac{1}{6}\sin(\pi \alpha)^2  = \frac{1}{6}\cos(\frac{\pi \alpha}{2})^2(5 + \cos(\pi \alpha)).
\end{equation}
For the other marginal, the unitarity of $\bE_{\alpha}$ is given by
\begin{equation}
    u(\bE_{\alpha}) = \frac{1}{3} \tr[T_{\bE,\alpha}^\dag T_{\bE,\alpha}] = \sin(\frac{\pi \alpha}{2})^4 + \frac{1}{6}\sin(\pi \alpha)^2  = \frac{1}{6}\sin(\frac{\pi \alpha}{2})^2(5 - \cos(\pi \alpha)).
\end{equation}
This completes the proof.
\end{proof}

\subsection{Analytical form for marginal unitarities of $CNOT_{AB}^\alpha$ isometry}\label{sec:cnotalpha}

\begin{lemma}\label{lemma:CNOT-marginal-unitarities}
    For the isometry $ \V(\rho)_{\alpha} := CNOT_{AB}^\alpha ( \rho \x \ketbra{0} )$ where $0 \leq \alpha \leq 1$, we define the marginals $\E_{\alpha} (\rho) := \tr_B[\V(\rho)_{\alpha}]$ and $\bE_{\alpha} (\rho) := \tr_A[\V(\rho)_{\alpha}]$. The unitarities of each marginal are
    \begin{equation}
        u(\E_{\alpha}) = 1- \frac{2s}{3}
    \end{equation}
    and
    \begin{equation}
        u(\bE_{\alpha}) = \frac{s}{3}
    \end{equation}
    respectively, where $s=\sin^2(\frac{\pi \alpha}{2})$.
\end{lemma}
If we consider the sum of the marginals from Lemma \ref{lemma:CNOT-marginal-unitarities} we have
\begin{equation}
    u(\E_{\alpha}) + u(\bE_{\alpha}) =  1- \frac{s}{3}
\end{equation}
with $0 \leq s \leq 1$.
\begin{proof}(of Lemma \ref{lemma:CNOT-marginal-unitarities})
    The proof follows in a similar way to the $SWAP^\alpha$ case.
    For the $CNOT_{AB}$ channel we use the notation $CNOT_{AB}(\rho) = U \rho U^\dag$, to clarify that we mean the unitary matrix $U$ itself. We can diagonalise $U$ w.r.t. the computational basis by applying a hadamard transform $H = (Z + X)/\sqrt{2}$ on the target qubit before and after such that the sandwiched unitary is the controlled phase gate:
    \begin{equation}
        \begin{split}
            U &= \frac{1}{2}(\I \x \I + Z \x \I + \I \x X - Z \x X),\\
            &= \frac{1}{2}(\I \x H \I H + Z \x H \I H + \I \x HZH - Z \x HZH), \\
            &= (\I \x H) \frac{1}{2} (\I \x \I + Z \x \I + \I \x Z - Z \x Z )(\I \x H), \\
            &= (\I \x H) (\ketbra{00} + \ketbra{01} + \ketbra{10} -\ketbra{11}) (\I \x H), \\
            &= (\I \x H) (\ketbra{00} + \ketbra{01} + \ketbra{10} + e^{i \pi}\ketbra{11}) (\I \x H).
        \end{split}
    \end{equation}
    Therefore we have
    \begin{equation}
        \begin{split}
            U^\alpha &= (\I \x H) (\ketbra{00} + \ketbra{01} + \ketbra{10} + e^{i \pi \alpha}\ketbra{11}) (\I \x H), \\
            &= \ketbra{0} \x \I + \frac{1}{2}(1 + e^{i \pi \alpha})(\ketbra{1} \x \I) + \frac{1}{2}(1 - e^{i \pi \alpha})(\ketbra{1} \x X).
        \end{split}
    \end{equation}
    From the isometry definition we have $\V(\rho)_{\alpha} = U^\alpha \rho \x \ketbra{0} (U^\alpha)^\dag $, substituting in the definition of $U^\alpha$ we can show that the unital block $T$ for the 1 qubit channels $\E_{\alpha}$ \& $\bE_{\alpha}$ will be
    \begin{equation}
        T_{\E,\alpha} =
        \bordermatrix{~ 
            & \vket{X/\sqrt{2}} & \vket{Y/\sqrt{2}} & \vket{Z/\sqrt{2}} \cr
            \vbra{X/\sqrt{2}} &   \cos^2(\frac{\pi \alpha}{2}) & \frac{1}{2}\sin(\pi \alpha) & 0   \cr
            \vbra{Y/\sqrt{2}}  & -\frac{1}{2}\sin(\pi \alpha) &  \cos^2(\frac{\pi \alpha}{2}) & 0 \cr
            \vbra{Z/\sqrt{2}}  & 0 & 0 &  1\cr
            },
    \end{equation}
    and
    \begin{equation}
        T_{\bE,\alpha} =
        \bordermatrix{~ 
            & \vket{X/\sqrt{2}} & \vket{Y/\sqrt{2}} & \vket{Z/\sqrt{2}} \cr
            \vbra{X/\sqrt{2}} & 0 & 0 & 0   \cr
            \vbra{Y/\sqrt{2}}  & 0 & 0 & 0 \cr
            \vbra{Z/\sqrt{2}}  & 0 & \frac{1}{2}\sin(\pi \alpha) &  \sin^2(\frac{\pi \alpha}{2}) \cr
            },
    \end{equation}
    respectively. Therefore, with some multiplication, the unitarity of $\E_{\alpha}$ is given by $u(\E_{\alpha})  = 1- \frac{2}{3} \sin^2(\frac{\alpha \pi}{2})$ and the unitarity of $\bE_{\alpha}$ is given by $ u(\bE_{\alpha}) = \frac{1}{3} \sin^2(\frac{\alpha \pi}{2})$.
    This completes the proof.
\end{proof}

\subsection{Analytical form for marginal unitarities of $CNOT_{BA}^\alpha \circ CNOT_{AB}$ isometry}\label{sec:cnotalpha-cnot}

\begin{lemma}\label{lemma:cnotalpha-cnot-marginal-unitarities}
    For the isometry $ \V(\rho)_{\alpha} := CNOT_{BA}^\alpha \circ CNOT_{AB} ( \rho \x \ketbra{0} )$ where $0 \leq \alpha \leq 1$, we define the marginals $\E_{\alpha} (\rho) := \tr_B[\V(\rho)_{\alpha}]$ and $\bE_{\alpha} (\rho) := \tr_A[\V(\rho)_{\alpha}]$. The unitarities of each marginal are
    \begin{equation}
        u(\E_{\alpha}) = \frac{1}{3}(1-s)
    \end{equation}
    and
    \begin{equation} 
        u(\bE_{\alpha}) = 1- \frac{2}{3}(1-s) 
    \end{equation}
    respectively, where $s=\sin^2(\frac{\pi \alpha}{2})$.
\end{lemma}
If we consider the sum of the marginals from Lemma \ref{lemma:cnotalpha-cnot-marginal-unitarities} we have
\begin{equation}
    u(\E_{\alpha}) + u(\bE_{\alpha}) =  1- \frac{1}{3}(1-s)
\end{equation}
with $0 \leq s \leq 1$.
\begin{proof}(of Lemma \ref{lemma:cnotalpha-cnot-marginal-unitarities})
    Proof follows in the same way as the previous two lemmas. From the previous lemma we can write the unitary matrix  for the channel $CNOT_{BA}(\rho):= U_{BA} \rho U_{BA}^\dag$ as
    \begin{equation}
        U_{BA} = (H \x \I) (\ketbra{00} + \ketbra{01} + \ketbra{10} + e^{i \pi}\ketbra{11}) (H \x \I),
    \end{equation}
    and therefore
    \begin{equation}
        \begin{split}
            U_{BA}^\alpha &= (H \x \I) (\ketbra{00} + \ketbra{01} + \ketbra{10} + e^{i \pi \alpha}\ketbra{11}) (H \x \I), \\
            &= \I \x \ketbra{0} + \frac{1}{2}(1 + e^{i \pi \alpha})(\I \x \ketbra{1}) + \frac{1}{2}(1 - e^{i \pi \alpha})(X \x \ketbra{1}).
        \end{split}
    \end{equation}
    The unitary matrix for the channel $CNOT_{AB}(\rho):= U_{AB} \rho U_{AB}^\dag$ is given in the Pauli basis as
    \begin{equation}
        U_{AB} = \frac{1}{2}(\I \x \I + Z \x \I + \I \x X - Z \x X),
    \end{equation}
    therefore from the isometry definition $\V(\rho)_{\alpha} =U_{BA}^\alpha  U_{AB} \rho \x \ketbra{0} U_{AB}^{\dag} (U_{BA}^\alpha)^\dag$ we can show that the unital block $T$ for the 1 qubit channels $\E_{\alpha}$ \& $\bE_{\alpha}$ will be
    \begin{equation}
        T_{\E,\alpha} =
        \bordermatrix{~ 
            & \vket{X/\sqrt{2}} & \vket{Y/\sqrt{2}} & \vket{Z/\sqrt{2}} \cr
            \vbra{X/\sqrt{2}} & 0 & 0 & 0   \cr
            \vbra{Y/\sqrt{2}}  & 0 & 0 & 0 \cr
            \vbra{Z/\sqrt{2}}  & 0 & -\frac{1}{2}\sin(\pi \alpha) &  \cos^2(\frac{\pi \alpha}{2}) \cr
            },
    \end{equation}
    and
    \begin{equation}
        T_{\bE,\alpha} =
        \bordermatrix{~ 
            & \vket{X/\sqrt{2}} & \vket{Y/\sqrt{2}} & \vket{Z/\sqrt{2}} \cr
            \vbra{X/\sqrt{2}} &   \sin^2(\frac{\pi \alpha}{2}) & -\frac{1}{2}\sin(\pi \alpha) & 0   \cr
            \vbra{Y/\sqrt{2}}  & \frac{1}{2}\sin(\pi \alpha) &  \sin^2(\frac{\pi \alpha}{2}) & 0 \cr
            \vbra{Z/\sqrt{2}}  & 0 & 0 &  1\cr
            },
    \end{equation}
    respectively. Therefore, with some multiplication, the unitarity of $\E_{\alpha}$ is given by $u(\E_{\alpha})  = \frac{1}{3} \cos^2(\frac{\alpha \pi}{2}) $ and the unitarity of $\bE_{\alpha}$ is given by $ u(\bE_{\alpha}) = 1 - \frac{2}{3} \cos^2(\frac{\alpha \pi}{2}) $.
    This completes the proof.

\end{proof}

\section{SPAM robust estimation through interleaved RB}\label{append:interleaveduRB}
\subsection{Interleaved unitarity protocol for $(\E,\bE)$ without noise}\label{append:interleaveduRB-noiseless}
We now give a sketch of the proof for Protocol \ref{protocol:inter-uRB-E_A}. Define the elements of the Clifford group on qubit $A$ to be $\{\U_{A,i}\}$. We define channel induced by averaging over many Clifford unitaries as
\begin{equation}
    \U_A := \frac{1}{N} \sum_i^N \C_{A,i}.
\end{equation}
In the case $\Lambda=\Lambda_C=id$, the circuit diagram representation of Protocol \ref{protocol:inter-uRB-E_A} is
\begin{equation}
    \begin{array}{c}
        \Qcircuit @C=1.2em @R=.7em {
             \lstick{\rho_A} &  \gate{\U_A} &  \multigate{1}{{\U_{AB}}}    &  \gate{\U_A}  & \rstick{M_A}  \qw \\
             \lstick{\ketbra{0}} & \gate{\D} & \ghost{{\U_{AB}}}   & \gate{\D}  \qw & \blacktriangleright  \qw \gategroup{1}{2}{2}{3}{1.0em}{--} \\
            & & \dstick{ \mathrm{Repeat} \ k-1 \ \mathrm{times} \ \ \ \ \ \ \ \ \ \ \ \ \ \ }
    }
    \end{array}
    \vspace{0.6cm}
\end{equation}
where $\blacktriangleleft$ indicates the channel preparing $\ketbra{0}$ and $\blacktriangleright$ the trace operation. As $\blacktriangleright \blacktriangleleft = \D$, the circuit reduces to
\begin{equation}
    \begin{array}{c}
        \Qcircuit @C=1.2em @R=.7em {
            \lstick{\rho_A} & \qw &  \gate{\U_A} &  \multigate{1}{{\U_{AB}}}    &  \gate{\U_A} & \qw  &\rstick{M_A}   \qw \\
            ~ & \blacktriangleleft & \qw & \ghost{{\U_{AB}}}   & \qw  & \blacktriangleright  \qw & ~  \gategroup{1}{2}{2}{6}{1.9em}{--} \\
            & & & \dstick{\mathrm{Repeat} \ k-1 \ \mathrm{times}}
    }
    \end{array}
    \vspace{0.6cm}
\end{equation}
which further reduces to
\begin{equation}
    \begin{array}{c}
        \Qcircuit @C=1.2em @R=.7em {
            \lstick{\rho_A} &  \gate{\U_A} &  \gate{\E}    &  \gate{\U_A}  &\rstick{M_A}   \qw  \gategroup{1}{2}{1}{4}{1.0em}{--} \\
            & & \dstick{\mathrm{Repeat} \ k-1 \ \mathrm{times}}
    }
    \end{array}
    \vspace{0.6cm}
\end{equation}
Which is exactly the right form for the circuit to estimate $u(\E)$. The decay parameter $e_1$ in Protocol \ref{protocol:inter-uRB-E_A} is exactly $u(\E)$ in this idealised case.

The protocol to estimate $u(\bE)$ is the same as for $u(\E)$ replacing $\U_{AB}$ with $\swap \circ \U_{AB}$. This follows from the fact that
\begin{equation}
    \bE (\rho) := \tr_A \circ \ {\U_{AB}} (\rho \x \ketbra{0}) = \tr_B \circ \swap \circ {\U_{AB}} (\rho \x \ketbra{0}).
\end{equation}
or as a circuit diagram
\begin{equation}
    \begin{array}{ccc}
        \Qcircuit @C=0.0em @R=1.0em {
            - \hspace{2.4mm}  & \gate{\bE} & \hspace{2.4mm} -  \qw \\
        }
    & ~ \hspace{0.4cm} ~ &
        \Qcircuit @C=1.0em @R=1.0em {
               & \multigate{1}{{\U_{AB}}} &  \qswap &  \qw  \\
            \blacktriangleleft  & \ghost{{\U_{AB}}} &  \qswap \qwx &  \qw \blacktriangleright \inputgroup{1}{2}{.75em}{=}  \\
        }
    \end{array}
\end{equation}
where we implicitly assume $d_A=d_B$, which is appropriate in this two qubit case.

\subsection{Efficient implementation of protocols}\label{sec:efficient-protocols}

In the experiments for the SPAM robust CUP-set, we perform an \textit{efficient} unitarity RB protocol, as introduced in \cite{dirkse2019efficient}. The protocol allows for rigorous bounds on the variance in the  associated decay curve, and therefore the value of unitarity extracted. We summarize the efficient unitarity RB protocol applied to our scheme here, where we consider $\E$ as a black box single qubit channel whose implementation is detailed in Appendix \ref{append:interleaveduRB-noiseless}.

\begin{protocol}[Efficient interleaved unitarity RB for $\E$.]
    \label{protocol:efficientC}
    \begin{enumerate}[wide, labelwidth=!, labelindent=0pt]
        \setlength{\itemsep}{2pt}
        \setlength{\parskip}{0pt}
        \setlength{\parsep}{0pt}
        \item \textbf{Select} a random sequence, $\U_{\textbf{k}} := \U_{k} \circ \E \circ \U_{k-1} \circ \E \circ ... \circ \U_2 \circ \E \circ \U_1$, of random Clifford gates interleaved with target channel $\E$.
        \item \textbf{Prepare} the system in the state $\rho_{\pm,i} := \frac{1}{d}(\ident \pm P_i)$ for all non-identity elements $P_i$, of the Pauli group $P \neq \ident$. In the single qubit input and output case, the states $\rho_{\pm,i}$ are pure states given by $\rho_{\pm,i} = \{ \ket{+}, \ket{-}, \ket{+i}, \ket{-i}, \ket{0}, \ket{1}\}$.
        \item \textbf{Estimate} the average purity of the sequence across all possible traceless input and output Pauli	:
        \begin{equation*}
            q_{\textbf{k}} = \frac{1}{d^2-1} \sum_{i,j}^{d^2-1} (\tr[P_j   \ \C_{\textbf{k}}(\rho_{+,i})] - \tr[P_j \ \C_{\textbf{k}}(\rho_{-,i})])^2.
        \end{equation*}
        \item \textbf{Repeat 1, 2 \& 3} for $N_k$ random sequences of length $k$, finding the average estimation $\mathbb{E}[q_{\textbf{k}}] := \frac{1}{N_k} \sum_{\textbf{k}}^{\textbf{N}_{\textbf{k}}} q_{\textbf{k}}$.
        \item \textbf{Repeat 1, 2, 3 \& 4} increasing the length of the sequence, e.g. $k = k + 1$.
        \item \textbf{Fit} the data with $\mathbb{E}[q_{\textbf{k}}] = c_1 s^{k-1} $, to find $s$ the estimated value of $u(\E)$. 
    \end{enumerate}
\end{protocol}

\subsection{Interleaved unitarity protocol for $(\E,\bE)$ with noise}\label{append:general-noise-in-protocol}
In this section we set out the minimum assumptions required to produce interleaved unitarity RB circuits, where all operations are assumed to be noisy. We then show how this effects the estimation of CUP-sets. For channels, states, functions, any $X$, we write the noisy version $\noisy{X}$.

For a two qubit system when we implement any gate or mid-circuit measurement, the noise associated with the process may effect the whole device. Therefore we should model errors as bipartite quantum channels. We make two simplifying assumptions about these errors. Firstly, we consider the noise to be fixed across the Clifford group gateset, such that $\Lambda_{\C,i} = \Lambda_\C$ for all $\U_i$. E.g. ${\U}_{i,N} =  \Lambda_\C \circ \U_i \x id_B$. Secondly, we assume the reset of a qubit is perfectly incoherent, but potentially noisy. Therefore the total channel can be written as ${\D}_{B,N} = \Lambda_{\D} \circ id_A \x \D_B$, with a general bipartite error channel $\Lambda_{\D}$.

A direct consequence of these two assumptions is that we can write the noisy version of the summation of Clifford unitaries ${\U}_{A,N}$ and the reset operation as
\begin{equation}
    \begin{array}{ccc}
    \Qcircuit @C=1.0em @R=1.0em {
        & \multigate{1}{\U_{A,N} \x \D_{B,N}} & \qw \\
        & \ghost{\U_{A,N} \x \D_{B,N}} &  \qw  \\
    }
    & ~ \hspace{0.4cm} ~ &
    \Qcircuit @C=1.0em @R=1.0em {
        & \gate{\U_{A}}  & \multigate{1}{\Lambda_C} & \qw \\
         & \gate{\D_{B}} & \ghost{\Lambda_C} &  \qw  \inputgroup{1}{2}{.75em}{=} \\
    }
    \end{array}
\end{equation}
where $\Lambda_C$ is an error channel associated with the operations together. Putting this together with a noisy version of the interleaved unitary $\U_{AB,N} = \U_{AB} \circ \Lambda_{AB}$ we can write a noisy version of the circuit for Protocol \ref{protocol:inter-uRB-E_A}:

\begin{equation}
    \begin{array}{c}
        \Qcircuit @C=1.2em @R=.7em {
             \lstick{\rho_{N}} & \gate{\U_A} & \multigate{1}{\Lambda_C}     &  \multigate{1}{{{\U_{AB}}}} &  \multigate{1}{\Lambda_{\U_{AB}}}  & \gate{\U_A}  & \multigate{1}{\Lambda_C}    & \rstick{M_{A,N}}  \qw \\
             \lstick{\noisy{\ketbra{0}}} & \gate{\D_B} & \ghost{\Lambda_C}   & \ghost{{{\U_{AB}}}} & \ghost{\Lambda_{\U_{AB}}} & \gate{\D_B}  & \ghost{\Lambda_C}   & \blacktriangleright  \qw \gategroup{1}{2}{2}{5}{1.0em}{--} \\
            & & & & \dstick{ \mathrm{Repeat} \ k-1 \ \mathrm{times} \hspace{4cm}}
    }
    \end{array}
    \vspace{0.4cm}
\end{equation}

Further we can write the reset operations as trace and preparation operations in our notation and absorb the initial and final error channels as SPAM errors in the $A$ subsystem. This leaves us with:
\begin{equation}
    \begin{array}{c}
        \Qcircuit @C=1.2em @R=.7em {
            \lstick{\rho_{N}} & \gate{\U_A}  & \multigate{1}{\Lambda_C}   &  \multigate{1}{{{\U_{AB}}}}   & \multigate{1}{\Lambda_{\U_{AB}}}  & \gate{\U_A}   &  \rstick{M_{A,N}}   \qw \\
           ~ & \blacktriangleleft & \ghost{\Lambda_C} & \ghost{{{\U_{AB}}}}   & \ghost{\Lambda_{\U_{AB}}}  & \qw \blacktriangleright  & ~  \gategroup{1}{2}{2}{6}{1.4em}{--} \\
        & & & \dstick{\mathrm{Repeat} \ k-1 \ \mathrm{times}}
    }
    \end{array}
    \vspace{0.2cm}
\end{equation}
As the protocol is SPAM robust, we get an	 estimation of the unitarity $u(\noisy{\E})$ of the channel $\noisy{\E}(\rho):=  tr_B \circ \Lambda_{AB} \circ \U_{AB} \circ \Lambda_C ( \rho \x \ketbra{0} )$.

When we implement the protocol for $\bE$ we will have an additional required operation, $\swap$, and the noise associated with it. We can write the noisy version of this in full generality, by including it with the preceding defined unitary ${\U_{AB}}$. E.g. $\noisy{\swap} \circ \U_{AB,N} := \Lambda_{S,AB} \circ \swap \circ \U_{AB}$. This leads to a noisy circuit of the form 
\begin{equation}
    \begin{array}{c}
        \Qcircuit @C=1.2em @R=.7em {
            \lstick{\rho_{N}} & \gate{\U_A}  & \multigate{1}{\Lambda_C}  &  \multigate{1}{{{\U_{AB}}}} &  \qswap &  \multigate{1}{\Lambda_{{S,AB}}}     & \gate{\U_A} & \multigate{1}{\Lambda_C}   &  \rstick{M_{A,N}} \qw \\
            \lstick{\noisy{\ketbra{0}}} & \gate{\D_B} & \ghost{\Lambda_C}   & \ghost{{{\U_{AB}}}} & \qswap \qwx    & \ghost{\Lambda_{{S,AB}}} & \gate{\D_B} & \ghost{\Lambda_C}  & \blacktriangleright  \qw \gategroup{1}{2}{2}{6}{1.2em}{--} \\
            & & & & & \dstick{\mathrm{Repeat} \ k-1 \ \mathrm{times} \hspace{4cm}}
    }
    \end{array}
    \vspace{0.4cm}
\end{equation}
Finally, absorbing the initial and final error channels as SPAM errors in the $A$ subsystem leaves us with:
\begin{equation}
    \begin{array}{c}
        \Qcircuit @C=1.2em @R=.7em {
            \lstick{\rho_{N}} & \gate{\U_A}   & \multigate{1}{\Lambda_C} &    \multigate{1}{{\U_{AB}}} &  \qswap  &  \multigate{1}{{{\Lambda_{S,AB}}}}   & \gate{\U_A}   &   \rstick{M_{A,N}}  \qw \\
           ~ & \blacktriangleleft & \ghost{\Lambda_C}  & \ghost{{\U_{AB}}} & \qswap \qwx  & \ghost{{{\Lambda_{S,AB}}}}   & \qw \blacktriangleright  & ~  \gategroup{1}{2}{2}{7}{1.6em}{--} \\
            & & & & \dstick{\mathrm{Repeat} \ k-1 \ \mathrm{times} \hspace{1cm}}
    }
    \end{array}
    \vspace{0.2cm}
\end{equation}
Giving an exact estimation of the unitarity $u(\noisy{\bE})$ for the channel $\noisy{\bE}(\rho) := tr_A \circ \Lambda_{S,AB} \circ \U_{AB} \circ \Lambda_C ( \rho \x \ketbra{0} ).$

\clearpage
\small
\bibliography{references.bib}
\end{document}